\documentclass[12pt]{article}

\usepackage{amsmath}
\usepackage{amsfonts}
\usepackage{latexsym}
\usepackage{graphicx}
\usepackage{amssymb}
\usepackage{amsthm}
\usepackage[margin=2cm]{geometry}
\usepackage{url}
\usepackage{color}
\usepackage{enumerate}
\usepackage[small]{caption}

\usepackage{tikz}
\usetikzlibrary{matrix,arrows}

\usepackage{lmodern}         % Polices vectorielles
\usepackage[utf8]{inputenc}  % Codage UNICODE (UTF-8)
\newtheorem{definition}{Definition}
\newtheorem{lemma}[definition]{Lemma}
\newtheorem{proposition}[definition]{Proposition}

\newtheorem{corollary}[definition]{Corollary}

\newtheorem{theorem}[definition]{Theorem}
\newtheorem{remark}[definition]{Remark}

\newcommand{\N}{\mathbb{N}}
\renewcommand{\P}{\mathbb{P}}
\newcommand{\Z}{\mathbb{Z}}

\newcommand{\R}{\mathbb{R}}

\newcommand{\Card}{{\rm Card\,}}
\newcommand{\Ker}{{\rm Ker}}
\newcommand{\UN}{(1,1,\ldots,1)}
\renewcommand{\a}{{\vec{a}}}
\renewcommand{\b}{{\vec{b}}}
\newcommand{\F}{{\mathcal F}}
\newcommand{\Fa}{\F_\a}
\newcommand{\Fb}{\F_\b}
\newcommand{\Fao}{\F_{\a,\omega}}
\newcommand{\KerF}{\Ker\, \Fa}

\newcommand{\G}{{\mathcal G}}
\renewcommand{\H}{{\mathcal H}}
\newcommand{\D}{{\mathcal D}}
\renewcommand{\P}{{\mathcal P}}
\renewcommand{\S}{{\mathcal S}}
\newcommand{\Ga}{\G_\a}
\newcommand{\Ia}{I_\a}

\newcommand{\Hao}{\H_{\a,\omega}}
\newcommand{\Iao}{I_{\a,\omega}}
\newcommand{\Gao}{\G_{\a,\omega}}
\newcommand{\Ha}{\H_\a}
\newcommand{\Hb}{\H_\b}

\newcommand{\bE}{{\mathbb E}}

\newcommand\til{\widetilde}
\newcommand{\Qzero}{{\mathcal Q}}
\newcommand{\flip}{\textsc{flip}}

\begin{document}

\title{A $d$-dimensional extension of Christoffel words\thanks{With the support of NSERC (Canada)}}
\author{{\sc S. Labb\'e and C. Reutenauer }\\  \\
\small LIAFA, Université Paris Diderot - Paris 7,\\ [-0.6ex]
\small Case 7014, 75205 Paris Cedex 13, France\\ [-0.6ex]
\small \tt labbe@liafa.univ-paris-diderot.fr\\[1.9ex]
\small Laboratoire de Combinatoire et d'Informatique Math\'ematique,\\ [-0.6ex]
\small Universit{\'e} du Qu{\'e}bec {\`a} Montr{\'e}al,\\[-0.6ex]
\small C. P. 8888 Succursale ``Centre-Ville'', Montr{\'e}al (QC), CANADA H3C  3P8\\[-0.6ex]
\small \tt reutenauer.christophe@uqam.ca
}
%\today
%\date{}
\date{\small Mathematics Subject Classifications: 05C75, 52C35, 68R15.}

\maketitle

\begin{abstract}
In this article, we extend the definition of Christoffel words to
directed subgraphs of the hypercubic lattice in arbitrary dimension that we
call Christoffel graphs. Christoffel graphs when $d=2$ correspond to
well-known Christoffel words.
Due to periodicity, the $d$-dimensional Christoffel graph can be embedded in a
$(d-1)$-torus (a parallelogram when $d=3$).
We show that Christoffel graphs have similar properties to those of
Christoffel words: symmetry of their central part and conjugation with their
reversal. Our main result extends Pirillo's theorem (characterization of
Christoffel words which asserts that a word $amb$ is a Christoffel word if and
only if it is conjugate to $bma$) in arbitrary dimension.
In the generalization, the map $amb\mapsto bma$ is seen as a flip operation on
graphs embedded in $\Z^d$ and the conjugation is a translation.  
We show that a fully periodic subgraph of the hypercubic lattice is a
translate of its flip if and only if it is a Christoffel graph.
\end{abstract}

%\tableofcontents

\section{Introduction}

This article is a contribution to the study of discrete planes and
hyperplanes in any dimension~$d$. We study only rational
hyperplanes, that is, those which are defined by an equation with rational
coefficients. We extract from such an hyperplane a finite pattern that we
call, for $d=3$, a \emph{Christoffel parallelogram}. We show that they are a
generalization of Christoffel words.

%DISCRETE PLANES
Discrete planes were introduced by \cite{Reveilles_1991} and further studied
\cite{debled-rennesson_reconnaissance_1995,MR1382845,A,MR1732895}.  Recognition
algorithms were proposed in \cite{MR1368203,MR2305655,MR1603656}.  See
\cite{MR2296869} for a complete review about many aspects of digital
planarity, such as characterizations in arithmetic geometry, periodicity,
connectivity and algorithms.
%GENERALIZED SUBSTITUTIONS and surface
Discrete planes can be seen as an union of square faces. Such stepped
surface, introduced in \cite{MR1279568,MR1247666} as a way to construct
quasiperiodic tilings of the plane, can be generated from multidimensional
continued fraction algorithms by introducing substitutions on square faces
\cite{MR1888763, MR1906478}.

While discrete planes are a satisfactory generalization of Sturmian words, it
is still unclear what is the equivalent notion of Christoffel words in higher
dimension. In \cite[Fig. 6.6 and 6.7]{fernique_these_2007}, fundamental
domain of rational discrete planes are constructed from the iteration of
generalized substitutions on the unit cube.  Recently \cite{MR3052947}
generalized central words to arbitrary dimension using palindromic closure. In
both cases the representation is nonconvex and has a boundary like a fractal.

In this article, we propose to extend the definition of Christoffel words to
directed subgraphs of the hypercubic lattice in arbitrary dimension that we
call Christoffel graphs. 
A similar construction, called \emph{roundwalk}, but serving a different
purpose was given in \cite{MR2074953} producing multi-dimensional words
that are closely related to $k$-dimensional Sturmian words.
Christoffel graphs when $d=2$ correspond to Christoffel words.
Due to its periods, the $d$-dimensional Christoffel graph can be embedded in a
$(d-1)$-torus and when $d=3$, the torus is a parallelogram.
This extension is motivated by Pirillo's theorem
which asserts that a word $amb$ is a Christoffel word if and only if it is
conjugate to $bma$.  In the generalization, the map $amb\mapsto bma$ is seen
as a flip operation on graphs embedded in $\Z^d$ and the conjugation is
replaced by some translation.  When $d=3$, our flip corresponds to a flip in a
rhombus tiling \cite{MR2856174,MR2440650,MR2330996}.
We show that these Christoffel graphs have similar properties to those
of Christoffel words: symmetry of their central part
(Lemma~\ref{lem:bodysymmetric}) and conjugation with their reversal
(Corollary~\ref{cor:conjugatetoreversal} and \ref{ChristParal}).
Our main result is Theorem~\ref{thm:pirilloforalld} which extends Pirillo's
theorem in arbitrary dimension.

We recall in Section~\ref{sec:christwords} the basic notion on Christoffel
words and discrete planes. The discrete hyperplane graphs are defined in
Section~\ref{sec:IH}.
The operation on them (flip, reversal and translation)
are introduced in Section~\ref{sec:fliprevtrans}.  We show that the flip of a
Christoffel graph is a translate of itself in
Section~\ref{sec:flipistranslating}. This is the sufficiency of the Pirillo's
theorem.  In Section~\ref{sec:Ddimpirillothm}, we consider the necessity and
obtain a $d$-dimensional Pirillo's theorem, our main result.  Finally, we
construct in the Section~\ref{sec:observ} in appendix, the mathematical
framework for the definition of discrete hyperplanes, since we could not find
explicit and complete references.

%\tableofcontents

\section{Christoffel words and discrete planes}\label{sec:christwords}

\subsection{Christoffel words}

Recall that Christoffel words are obtained by discretizing a line segment in the plane as follows: let $(p,q) \in \N^2$
with $\gcd(p,q) = 1$, and let $S$  be the line segment with endpoints $(0,0)$
and $(p,q)$.
\begin{figure}[ht]
\begin{center}
\includegraphics[width=3.50in]{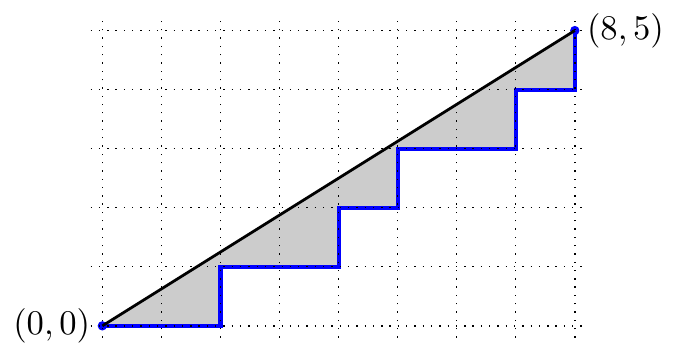}
\end{center}
\caption{The lower Christoffel word $w = aabaababaabab$.}
\label{F:christoffel}
\end{figure}
The word $w$ is a \emph{lower Christoffel
word} if the path induced by $w$ is under $S$ and if they both delimit a polygon with no
integral interior point. An \emph{upper Christoffel word} is defined similarly, by taking the path which is above the segment.
A \emph{Christoffel word} is a lower Christoffel word.
See Figure~\ref{F:christoffel} and reference \cite{BLRS08}.
An astonishing result about Christoffel words is the following characteristic
property given by Pirillo \cite{MR1854493}.

\begin{theorem}[Pirillo]\label{thm:pirillo}
A word $w=amb\in\{a,b\}^*$ is a Christoffel word if and only if $amb$ and
$bma$ are conjugate.
\end{theorem}

It is even known that the two words $amb$ and $bma$ are conjugate by palindromes \cite{MR1483437} Theorem 3.1 (see also \cite{MR2197281} Proposition 6.1): for example, the Christoffel word in Figure~\ref{F:christoffel} can be
factorized as a product of two palindromes, 
but also as a letter, a {\em central
word} $m$ and a last letter:
\[
w  = aabaa\cdot babaabab
   = a\cdot abaababaaba\cdot b
   = amb,
\]
and the conjugate word $w'$ of $w$ obtained by exchanging of the two
palindromes can also be factorized as the product of a letter, the same central word
$m$ and a last letter:
\[
w' = babaabab\cdot aabaa
   = b\cdot abaababaaba\cdot a
   = bma.
\]
%It turns out that the central word in both case is the same $m=m'$ and is a
%palindrome. 
%These observations are true in general.
Centrals words are the words $m$ such that $amb$ is a Christoffel word. They can be defined independently of Christoffel words:
a word $m$ is a \emph{central word} if and only if for some coprime integers $p$ and $q$, the length of $m$ is $p+q-2$ and $p$ and $q$
are periods of $m$.
In this case, the Christoffel word $amb$ is associated as above to the vector $(p,q)$.  See \cite{carpi_central_2005} for more informations
and \cite{berstel_sturmian_2007} for fourteen different characterizations of
central words.
% In the example, $q=5$ and $p=8$ are periods of the central word $m$. Moreover
% $p=8$ is a period of $am$ and $q=5$ is a period of $mb$. This is true in general.
There are also some properties which are satisfied by 
Christoffel words but do not characterize them.
\begin{lemma}\label{lem:periodcentral}
Let $w=amb$ be a Christoffel word of vector $(p,q)$. Then,
\begin{enumerate}[\rm(i)]
  \item the central word $m$ is a palindrome: $\til{m}=m$;
  \item $p$ is a period of $am$ and $q$ is a period of $mb$;
  \item the reversal $\til{w}$ of a Christoffel word $w$ is conjugate to $w$.
\end{enumerate}
\end{lemma}
The proof of (iii) follows from (i) and from Theorem~\ref{thm:pirillo}.
Words conjugate to their reversal were studied in \cite{bhnr}, are product of
two palindromes and are not necessarily Christoffel words. Moreover, not
every palindrome is a central word. In this article, we generalize
Theorem~\ref{thm:pirillo} to dimension 3. We also show that properties
like the one enumerated in Lemma~\ref{lem:periodcentral} hold.

\subsection{Discrete planes}

Given $\a=(a_1,a_2,a_3)\in\R^3$ and $\mu, \omega\in\R$, the
\emph{lower arithmetical discrete plane} \cite{Reveilles_1991}
$\P$ is the set of point
$x=(x_1,x_2,x_3)\in\Z^3$ satisfying
\[
\mu \leq a_1x_1 + a_2x_2 + a_3x_3 < \mu + \omega.
\]
The parameter $\omega$ is called the (arithmetic) \emph{width}.
If $\omega = \Vert \a \Vert_1 = |a_1|+|a_2|+|a_3|$, then the discrete plane is
said to be \emph{standard}.
Standard arithmetical discrete plane
can be furnished
with a canonical structure of a two-dimensional, connected, orientable
combinatoric manifold without boundary, whose faces are quadrangles and whose
vertices are points on the plane \cite{MR1382845}. 
See the appendix in Section~\ref{sec:observ} where we provide the mathematical
framework for the definition of discrete hyperplanes $\P$ and
\emph{stepped surfaces} $\S$ \cite{MR2330996}.

Let $k$ be an integer such that $0\leq k<d$.
We say that $u,v\in\Z^d$ are \emph{$k$-neighbor} if and only if
\[
\Vert v-u \Vert_\infty = 1
\text{ and }
\Vert v-u \Vert_1 \leq d - k.
\]
In this article, we are interested in the graph representing
the $2$-neighboring relation for the discrete plane  $\P$
and in general the
$(d-1)$-neighboring relation for the discrete hyperplane in $\Z^d$.
Note that $u,v\in\Z^d$ are $(d-1)$-neighbors if and only if their difference is
$\pm e_i$ for some $i$ such that $1\leq i\leq d$.

\section{Discrete hyperplane graphs}\label{sec:IH}

Let $a_1,\ldots,a_d$ be relatively prime positive integers and
$s=\Vert\a\Vert_1=\sum a_i$ be their sum.
We denote $\a=(a_1,a_2,\ldots,a_d) \in \N^d$.
We define the mapping $\Fa:\Z^d\to \Z / s\Z$ sending each integral vector $(x_1,\ldots,x_d)$ onto $\sum_ia_ix_i \bmod s$. 
We identify $\Z/s\Z$ and $\{0,1,\ldots, s-1\}$. A total order on $\Z/s\Z$
is defined correspondingly; it is this order that is used in the definition of
$\Ha$ below. 
The map $\Fa$ induces a $\Z^d$-action $x\cdot g = g + \Fa(x)$ on the cyclic
group $\Z/s\Z$, so that it is a rational case of the $\Z^2$-action on the
torus as studied in \cite{MR1782038,MR1906478}.
%Let $\a=(a_1,a_2,\ldots,a_d) \in \Z^d$ be some normal vector such that
%$gcd(a_1,a_2,\ldots,a_d)=1$. We suppose $a_i>0$ for all $i$ with $1\leq i \leq
%d$.  Let $s$ be the sum $s=\Vert\a\Vert_1=a_1+a_2+\cdots+a_d$.  We define the
%mapping
%\[
%\begin{array}{llll}
%\Fa: & \Z^d    & \to      & \Z / s\Z \\
%   & (x_1,x_2,\cdots,x_d) & \mapsto  & a_1x_1+a_2x_2+\cdots+a_dx_d \bmod s
%\end{array}
%\]
%
We consider $\bE_d = \{ (u,u+e_i) : u\in\Z^d \,\text{and}\, 1\leq i\leq d\}$,
the set of oriented edges of the hypercubic lattice.
Note that the set $\bE_d$ also corresponds to the Cayley graph of $\Z^d$ with
generators $e_i$ for all $i$ with $1\leq i\leq d$.

\subsection{The Christoffel graph $\Ha$}

The \emph{Christoffel graph} $\Ha$ of normal vector $\a$ is the subset of
edges of $\bE_d$ increasing for the function $\Fa$:
\[
\Ha = \{ (u, u+e_i) \in \bE_d : \Fa(u) < \Fa(u+e_i) \}.
\]
An example of the graph $\Ha$ when $d=2$ and $\a=(a_1,a_2)=(2,5)$ is shown in
Figure~\ref{fig:graph25} (left) where the edges are represented in blue and a
small red circle surrounds the origin.
\begin{figure}[h]
\begin{center}
\includegraphics{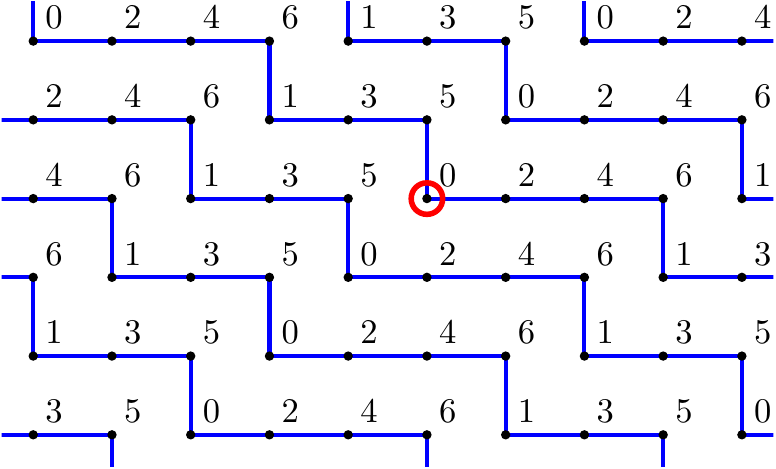}
\quad\quad
\includegraphics{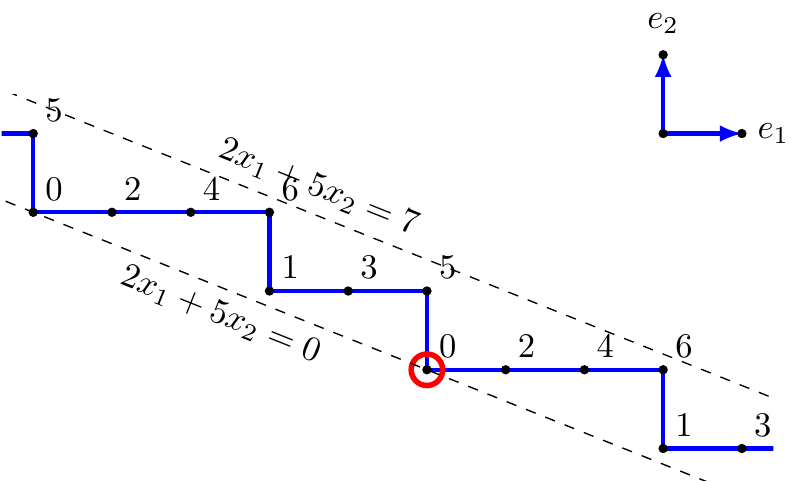}
\end{center}
\caption{Left: the graph $\Ha$ with $\a=(2,5)$.
Right: Standard discrete line $\P$ of normal vector $\a=(2,5)$.}
\label{fig:graph25}
\end{figure}
A first observation is stated in the next lemma.

\begin{lemma}
The graph $\Ha$ is invariant under the translation by the vector $\sum_{i=1}^d e_i=\UN$.
\end{lemma}

The proof is postponed at Lemma~\ref{lem:invariantKerF} where we show that
the graph $\Ha$ is invariant under all translations $t\in\KerF$.
Because of this invariance, a question is to find a good representent for
the equivalence class $x+\UN\Z$ for each $x\in\Z^d$. 
It is natural to choose $\bar x\in x+\UN\Z$ such that
\begin{equation}\label{eq:standarddiscreteplane}
0\leq \sum a_i \bar x_i < s.
\end{equation}
If $(u,v)$ is an edge of $\Ha$ such that $v-u=e_i$ then
$(\bar u, \bar v)$ 
is a pair of points which are (d-1)-neighbors
satisfying
Equation~\eqref{eq:standarddiscreteplane} and
$\bar v-\bar u=e_i$.
Thus the vertices satisfying Equation~\eqref{eq:standarddiscreteplane}
are a set of representents for the vertices of $\Ha$, see
Figure~\ref{fig:graph25} (right). 
Thus, each connected component of the graph $\Ha$ corresponds exactly to a
\emph{standard discrete plane} $\P$ with the $(d-1)$-neighbor relation.
The advantage of $\Ha$ over the discrete hyperplane $\P$ is its algebraic
structure. The next lemma gives an equivalent definition of the edges of the
graph $\Ha$. It will be useful in the sequel.

Let $a,b\in[0,s[$ be two integers.
If $a<b$ then $]a,b]$ is a subinterval of $[0,s[$.
If $a>b$ then $]a,b]=]a,s[\,\cup\,[0,b]$ is defined as the union of two subintervals of $[0,s[$.

\begin{lemma}\label{lem:edgesinterval}
Let $(u,v)\in\bE_d$ such that $v=u+e_i$ for some $1\leq i\leq d$. Then,
\begin{align}
(u,v)\in \Ha
\iff \Fa(u)\in[0,s-a_i-1]
\iff \Fa(v)\in[a_i,s-1]
\iff 0\notin\,]\Fa(u),\Fa(v)],\\
%\iff \Fa(u)<\Fa(v).
(u,v)\notin \Ha
\iff \Fa(u)\in[s-a_i,s-1]
\iff \Fa(v)\in[0,a_i-1]
\iff 0\in\,]\Fa(u),\Fa(v)].
%\iff \Fa(u)\geq\Fa(v).
\end{align}
\end{lemma}

For each permutation $\sigma$ of the set $\{1,2,\cdots,d\}$,
there exists a \emph{$\sigma$-path}
\[
(u,u+e_{\sigma(1)}),
(u+e_{\sigma(1)}, u+e_{\sigma(1)}+e_{\sigma(2)}),
\cdots,
(u+\sum_{i=1}^{d-1} e_{\sigma(i)}, u+\sum_{i=1}^d e_{\sigma(i)})
\]
made of $d$ edges of $\bE_d$
going from the vertex $u\in\Z^d$ to the vertex $u+\sum_{i=1}^d e_i$.

\begin{lemma}\label{lem:allbutone}
All of the $d$ edges of a $\sigma$-path but one belong to $\Ha$.
\end{lemma}

\begin{proof}
% To each edge $(u,v)\in\bE_d$, we associate the interval $]\Fa(u),\Fa(v)]$ if
% $(u,v)\in\Ha$ and the union of two interval
% $]\Fa(u),\Fa(v)]=\,]\Fa(u),s[\,\cup\,[0,\Fa(v)]$ if
% $(u,v)\notin\Ha$.
To each edge
$
(u+\sum_{i=1}^{k-1} e_{\sigma(i)}, u+\sum_{i=1}^k e_{\sigma(i)})
$
corresponds an interval (or an union of two intervals according to the above
remark)
$
[\Fa(u)+\sum_{i=1}^{k-1} a_{\sigma(i)}, \Fa(u)+\sum_{i=1}^k a_{\sigma(i)}[
$. 
Since $\sum_{i=1}^{d} a_{\sigma(i)} = \sum_{i=1}^{d} a_{i}=s$,
those $d$ sets, with $1\leq k\leq d$, are a partition of $[0, s[$.
Therefore, only one of them contains $0$.
From Lemma~\ref{lem:edgesinterval}, only one edge of the $\sigma$-path do not
belong to $\Ha$.
\end{proof}

Let $R\subseteq\{1,2,\cdots,d\}$ and $u\in\Z^d$. An \emph{hypercube graph from
vertex $u$ to vertex $u+\sum_{i\in R}e_i$} with $2^{\Card R}$ vertices is the
subgraph of $\bE_d$ defined by
\[
    \left\{\left(u+\sum_{i\in P}e_i, u+\sum_{i\in Q}e_i\right)\in\bE_d
	\,\middle|\,
	P\subset Q\subseteq R
	\text{ and }
    \Card Q\setminus P = 1\right\}.
\]
Each nonedge of $\Ha$ implies the presence of a hypercube graph with
$2^{d-1}$ vertices orthogonal and incident to it. For example, $(-e_1,0)\notin\Ha$ and
$(0,e_2)\in\Ha$ when $\a=(2,5)$. This is proved in the next lemma.

\begin{lemma}\label{lem:hypercube}
If $(u,v)\in\bE_d\setminus\Ha$,
then the hypercube graph from vertex $v$ to vertex $u+\sum_{i=1}^de_i$ with
$2^{d-1}$ vertices is a subgraph of $\Ha$.
\end{lemma}

\begin{proof}
From Lemma~\ref{lem:allbutone},
the last $d-1$ edges of every $\sigma$-path starting with the edge
$(u,v)$ and ending in $u+\sum_{i=1}^d e_i$ are in $\Ha$.
The set of last $d-1$ edges of these paths generates an hypercube graph
from vertex $v$ to vertex $u+\sum_{i=1}^de_i$.
\end{proof}

% \begin{proof}
% Let $(u,v)\in\bE_d\setminus\Ha$ with $v=u+e_k$.
% Then by Lemma~\ref{lem:edgesinterval}, $\Fa(v)\in[0,a_k-1]$.
% Let $S$ and $T$ be sets such that $S\subset T\subseteq\{1, 2, \cdots,
% d\}\setminus\{k\}$ and $T$ has one element that $S$ does not have, i.e.,
% $T\setminus S=\{\ell\}$ with $1\leq\ell\leq d$.
% We want to show that
% \[
% \left(v+\sum_{i\in S}e_i, v+\sum_{i\in T}e_i\right) \in \Ha
% \]
% for each sets $S$ and $T$.
% We have
% \[
%     \Fa\left(v+\sum_{i\in T}e_i\right)
%     =\Fa(v)+\sum_{i\in T}a_i
%     \in\left[\sum_{i\in T}a_i,a_k-1+\sum_{i\in T}a_i\right]
%     \subseteq\left[a_\ell,s-1\right].
% \]
% We conclude by Lemma~\ref{lem:edgesinterval}.
% \end{proof}

A line containing some point $x\in\Z^d$ parallel to $e_i$ in
the hypercubic lattice $\bE_d$ is a set
\[
L_{x,i} = \{ (x+ke_i, x+(k+1)e_i): k\in\Z \}\subset\bE_d.
\]
The intersection $L_{x,i}\cap\Ha$ of such a line with a discrete hyperplane
graph $\Ha$ is made of consecutive edges and nonedges.
The next Lemma states that Christoffel words appear in this sequence.
\begin{lemma}
The sequence of consecutive edges and nonedges in $L_{x,i}\cap\Ha$ is
periodic and the period is a Christoffel word.
\end{lemma}
\begin{proof}
Each subset $L_{x,i}\cap\Ha$ can be described by the subgroup of $\Z/s\Z$
generated by $\Fa(e_i)$, i.e.,
\[
(x+ke_i, x+(k+1)e_i)\in L_{x,i}\cap\Ha
\iff
0\leq \Fa(x+ke_i) < s-a_i
\]
This corresponds to the well-known construction of Christoffel words from
the labelling of Cayley graphs of $\Z/s\Z$ with the generator $a_i$ 
\cite[Section 1.2 Cayley graph definition]{BLRS08}.
\end{proof}

For example, in the discrete hyperplane graph $H_{(2,5)}$ shown in
Figure~\ref{fig:graph25}, coding an edge by the letter $a$ and a nonedge by
letter $b$, we get the periods $aaabaab$ and $abbabbb$ 
% $aaaab$, $aaabaabaab$ and $ab$ 
for the lines $L_{x,i}\cap\Ha$ for $i=1,2$ respectively.
Both are Christoffel words.
% $aaaab$, $aaabaabaab$ and $ab$ for the lines
% $L_{x,i}\cap\Ha$ for $i=1,2,3$ respectively.

\begin{definition}[Image]
Let $f:\Z^d\to S$ be an homomorphism of $\Z$-module.
For some subset of edges $X\subseteq\bE_d$, we define
the image by $f$ of the edges $X$ by
\[
    f(X) = \{(f(u), f(v)) \mid (u,v)\in X\}.
\]
\end{definition}
This definition allows to define the graphs $\Ia$ and $\Ga$ as
projections of $\Ha$ in the sections below.

% \todo{
% We have
% $\bE'_d=\pi(\bE_d)$,
% $\Ia=\pi(\Ha)$,
% $\bE''_d=\Fa(\bE_d)$ and
% $\Ga=\Fa(\Ha)$.
% }

\subsection{The graph $\Ia$}\label{sec:Ia}

Let $\pi$ be the orthogonal projection from $\R^d$ onto the hyperplane $\cal
D$ of equation $\sum x_i=0$. Its restriction to the stepped surface $\S$ of
the discrete plane $\P$ of normal vector $\a\in\Z^d$ is a
bijection onto $\D$. It maps $\P$, the integral points in $\S$, onto a lattice $L$
\cite[section 2.2]{MR1906478} in $\D$ spanned by the vectors $h_i=\pi(e_i)$; they
satisfy $\sum_i h_i=0$. Note
that $\pi(\Z^d)$ is also equal to $L$, since each point in $\mathbb
Z^d$ is congruent to some point in $\P$ modulo the kernel of the
projection. We may identify the set $L$ and $\Z^d/(1,1,\cdots,1)\Z$,
since two integral points are projected by $\pi$ onto the same point if and
only if their difference is a multiple of the vector $(1,1,\ldots,1)$ and
since this multiple is necessarily an integral multiple.
Since $\Fa\UN=0$, the mapping $\Fa$ induces a mapping $\Fa':L\to\Z/s\Z$.
We have the following commuting diagram:
\begin{center}
\begin{tikzpicture}[auto]
\node (A) at (0,0) {$\Z^d$};
\node (B) at (6,0) {$\Z/s\Z$};
\node (C) at (3,-2) {$L=\Z^d/\UN\Z$};
\draw (A) edge[->] node {$\Fa$} (B);
\draw (C) edge[->] node {$\Fa'$} (B);
\draw (A) edge[->] node {$\pi$} (C);
\end{tikzpicture}
\end{center}

We consider the directed graph whose set of edges is
$\Ia=\pi(\Ha)$.
% We consider the directed graph whose set of vertices is $L$ and whose set of edges is 
% $\bE'_d=\{(l,l+h_i): l\in L \,\text{and}\, 1\leq i\leq d\}$. We may identify the edge $(l,l+h_i)$ with the segment $[l,l+h_i]$ in the hyperplane $\D$; thus the 
% graph is naturally embedded in $\D$.
% Now we define the directed subgraph $\Ia\subset\bE'_d$ 
% consisting of the edges 
% \[
% \Ia = \{ (l, l+h_i) \in \bE'_d : \Fa'(l) < \Fa'(l+h_i) \}.
% \]
The graphs $\Ia$ 
for $\a=(a_1,a_2)=(2,5)$ and
$\a=(a_1,a_2,a_3)=(2,3,5)$
are shown in Figure~\ref{fig:I25_I235}. Note that the orientation of an edge is
redundant when $d=3$, since each edge is oriented as one of the vector $h_i$. 
\begin{figure}[h!]
\begin{center}
\includegraphics[width=0.44\linewidth]{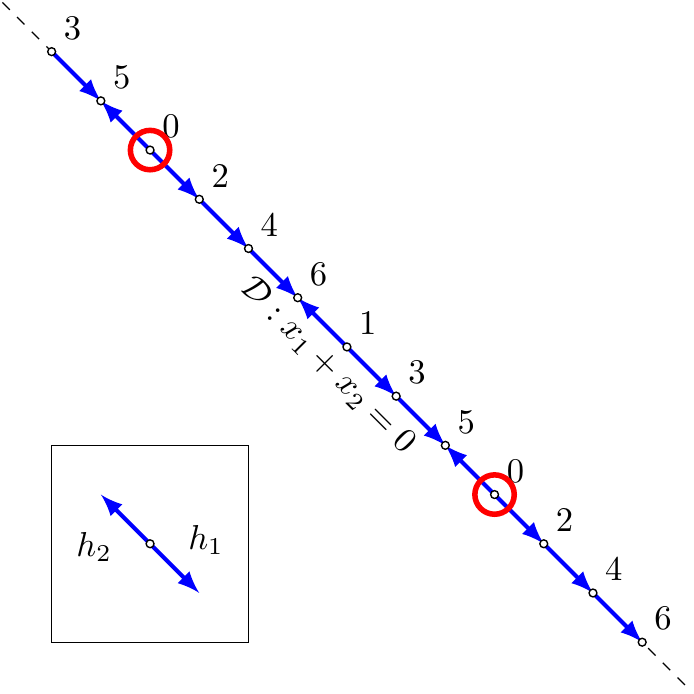}
\includegraphics[width=0.52\linewidth]{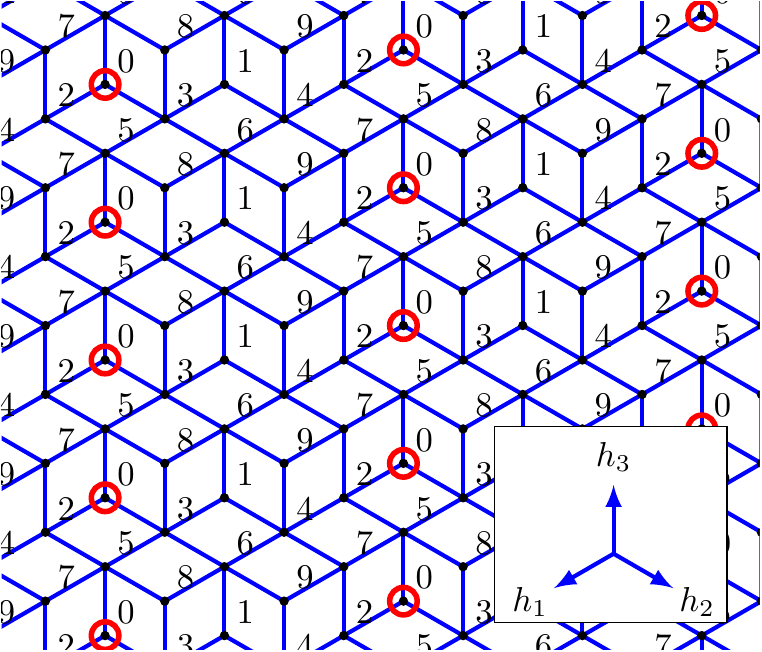}
\end{center}
\caption{Left: the graph $\Ia$ when $\a=(2,5)$. Right: the graph $\Ia$ when $\a=(2,3,5)$. 
The label at each vertex is its image under $\Fa'$.}
\label{fig:I25_I235}
\end{figure}

Lemma~\ref{lem:hypercube} can be seen on $\Ia$ when $d=3$ by the fact that
each nonedge is the short diagonal of a rhombus.
For example, if $\a=(2,3,5)$ then $(-h_1,0)\notin\Ia$. From the lemma, the
paths $(0,h_2), (h_2,h_2+h_3)$ and $(0,h_3),(h_3,h_2+h_3)$ are in $\Ia$.

The next lemma is a generalization of the fact that $\Ia$ is a tiling of
rhombus when $d=3$ proved in \cite{MR1382845} and \cite{MR1906478}.
Indeed, each rhombus is the projection under $\pi$ of one of three types of
square in $\R^3$.
Below,
%by \emph{$k$-hypercube} in $\R^d$ we mean the convex hull in $\R^d$ of
%dimension $k$.
the projection under $\pi$ of the convex hull of the
$2^k$ vertices of a $k$-dimensional hypercube graph in $\bE_d$
is called a \emph{$k$-dimensional parallelotope}. 
The edges of such a parallelotope have equal length.
When $d=3$, a
$(d-1)$-dimensional parallelotope is a rhombus.

\begin{proposition}
    The graph 
    $\Ia$ produces a tiling of $\D$ by 
    $d$ types of
    $(d-1)$-dimensional parallelotopes.
\end{proposition}

\noindent
In the following proof the fractional part of a real number $x\in\R$ is
denoted by $\{x\}=x-\lfloor x\rfloor$.

%J'ai lu l'exemple au début de la preuve; tu montres que x est un barycentre de
%\pi(u), \pi(u)+h_1, \pi(u)+h_1+h_2, \pi(u)=h_1+h_2+h_3, c'est bien ça? donc x
%est dans l'enveloppe convexe de ces 4 points.

\begin{proof}
Each real point $x=(x_1,\cdots,x_d)\in\R^d$ of the hyperplane $\D$ is contained in
a $(d-1)$-simplex with vertices $\{\pi(u)+\sum_{i=1}^k h_{\sigma(i)} : 0\leq k\leq
d-1\}$ for $u=(\lfloor x_1\rfloor,\cdots,\lfloor x_d\rfloor)\in\Z^d$ 
and permutation $\sigma$ of $\{1,2,\cdots,d\}$ such that
$\{x_{\sigma(1)}\}\geq\{x_{\sigma(2)}\}\geq\cdots\geq\{x_{\sigma(d)}\}$.
%Indeed, we have
%\begin{align*}
%x &= u + \sum_{i=1}^d \{x_{\sigma(i)}\}e_{\sigma(i)}
%  = u + \sum_{i=1}^{d-1}\left(\{x_{\sigma(i)}\}-\{x_{\sigma(i+1)}\}\right) 
%         \left(e_{\sigma(1)}+\cdots+e_{\sigma(i)}\right)
%       + \{x_{\sigma(d)}\}\sum_{i=1}^de_{\sigma(i)}\\
%  &= \left(1-\{x_{\sigma(d)}\}\right) u
%    + \sum_{i=1}^{d-1}\left(\{x_{\sigma(i)}\}-\{x_{\sigma(i+1)}\}\right) 
%         \left(u+e_{\sigma(1)}+\cdots+e_{\sigma(i)}\right)
%       + \{x_{\sigma(d)}\}u+\sum_{i=1}^de_{\sigma(i)}
%\end{align*}
%Therefore,
%\[
%x = \pi(x)
%  = \pi(u) + \sum_{i=1}^{d-1}\left(\{x_{\sigma(i)}\}-\{x_{\sigma(i+1)}\}\right) 
%         \left(h_{\sigma(1)}+\cdots+h_{\sigma(i)}\right)
%       + \{x_{\sigma(d)}\}\sum_{i=1}^d h_{\sigma(i)}.
%\]
We illustrate this on an example.
Suppose $d=4$ and $x$ is such that $\sigma$ is the identity permutation on
$\{1,2,3,4\}$.
We have
\begin{align*}
x &= 
\left\{
\begin{array}{l}
u\\
+\{x_1\}e_1\\
+\{x_2\}e_2\\
+\{x_3\}e_3\\
+\{x_4\}e_4\\
\end{array}
\right.
  = 
\left\{
\begin{array}{l}
u\\
+\left(\{x_1\}-\{x_2\}\right) e_1\\
+\left(\{x_2\}-\{x_3\}\right) (e_1+e_2)\\
+\left(\{x_3\}-\{x_4\}\right) (e_1+e_2+e_3)\\
+\{x_4\}(e_1+e_2+e_3+e_4)
\end{array}
\right.
  = 
\left\{
\begin{array}{l}
\left(1-\{x_1\}\right)u\\
+\left(\{x_1\}-\{x_2\}\right) (u+e_1)\\
+\left(\{x_2\}-\{x_3\}\right) (u+e_1+e_2)\\
+\left(\{x_3\}-\{x_4\}\right) (u+e_1+e_2+e_3)\\
+\{x_4\}(u+e_1+e_2+e_3+e_4).
\end{array}
\right.
\end{align*}
Therefore
$x$ is in the convex hull of the $3$-simplex with vertices
$\{\pi(u), \pi(u)+h_1, \pi(u)+h_1+h_2, \pi(u)+h_1+h_2+h_3\}$ since
\[
x = \pi(x) =
\left\{
\begin{array}{l}
 (1+\{x_4\}-\{x_1\})\pi(u)\\
+(\{x_1\}-\{x_2\}) (\pi(u)+h_1)\\
+(\{x_2\}-\{x_3\}) (\pi(u)+h_1+h_2)\\
+(\{x_3\}-\{x_4\}) (\pi(u)+h_1+h_2+h_3).
\end{array}
\right.
\]
Consider the $\sigma$-path in $\bE_d$ starting at vertex $u$ and ending at
$u+\sum_{i=1}^d e_i$.
From Lemma~\ref{lem:allbutone}, there is an edge $(u',v')$ of the
$\sigma$-path that is not in $\Ha$.
From Lemma~\ref{lem:hypercube},
the hypercube graph from vertex $v'$ to vertex $u'+\sum_{i=1}^d e_i$ with
$2^{d-1}$ vertices is a subgraph of $\Ha$.
Therefore, the hypercube graph (projected in $\pi(\bE_d)$) going from vertex
$\pi(v')$ to vertex $\pi(u')$ with $2^{d-1}$ vertices is a subgraph of $\Ia$.
The convex hull of this graph contains the $(d-1)$-simplex with vertices
$\{\pi(u)+\sum_{i=1}^k h_{\sigma(i)} : 0\leq k\leq d-1\}$ which in turns
contains $x$.
Therefore the point $x$ of the hyperplane $\D$ is contained in the image under
$\pi$ of the convex hull of a $(d-1)$-dimensional hypercube graph in $\bE_d$.
For almost every $x$, the inequalities
$\{x_{\sigma(1)}\}>\{x_{\sigma(2)}\}>\cdots>\{x_{\sigma(d)}\}$ are strict and
the parallelotope is unique.
We conclude that $\D$ is tiled by $(d-1)$-dimensional parallelotopes.
\end{proof}

%\todo{mention somewhere that when $d\geq 3$, then the orientation of the }\\
%\todo{edges in the projection is superfluous, which is not the case when $d=2$.}

\subsection{Kernel of $\Fa$ and $\Fa'$}

\begin{lemma}\label{lem:invariantKerF}
The discrete hyperplane graph $\Ha$ is invariant under any translation $t\in\KerF$.
\end{lemma}

%\todo{improve the proof using intervals instead of inequalities}

\begin{proof}
Let $u\in\Z^d$ and $t\in\KerF$.
We have $\Fa(u+t)=\Fa(u)+\Fa(t)=\Fa(u)$.
From Lemma~\ref{lem:edgesinterval},
$(u,u+e_i)\in \Ha$ if and only if $\Fa(u)\in[0,s-a_i-1]$
if and only if
$\Fa(u+t)\in[0,s-a_i-1]$
if and only if
$(u+t,u+e_i+t)\in\Ha$.
\end{proof}

We can find generators of the kernel of $\Fa$ when $d=3$. 

\begin{proposition}
If $d=3$, the kernel of $\Fa$ is
\[
\KerF = \langle(a_3, 0, -a_1), (0,a_3,-a_2),(a_2, -a_1, 0), (1,1,1) \rangle.
\]
\end{proposition}

The result is based on the
following well-known lemma.

\begin{lemma}\label{lem:index}
%\todo{find proper citation}
%\todo{make sure the statement is ok}
Let $K$ be a subgroup of $\Z^n$ generated by the rows of a $s\times n$
matrix $M\in\Z^{s\times n}$ of rank $n$. The index $[\Z^n:K]$ is equal to the $\gcd$ of the $n$-minors of
the matrix $M$.
\end{lemma}

\begin{proof}
By the rank condition, we have $s\geq n$. Suppose first that $M$ is in diagonal form; that is, the diagonal elements of $M$ are $d_1,\ldots,d_n$ and that the other elements are 0; by the rank condition, the $d_i$ are all nonzero. Then the subgroup is $K={d_1}\Z\times\cdots\times {d_n}\Z$, the quotient group is $\Z/{d_1}\Z\times\cdots\times\Z/d_n\Z$, and therefore the index is $d_1\cdots d_n$. Moreover the only nonzero $n$-minor is $d_1\cdots d_n$.

In the general case, it is well-known that the matrix $M$ may be brought into diagonal form by row and column operations within $\Z^{s\times n}$; moreover, these operations do not change the subgroup, up to change of basis in $\Z^{ n}$; and finally, the $\gcd$ of the $n$-minors is invariant under these operations. Thus the general case follows from the diagonal case.
\end{proof}

\begin{proof} (of the proposition)
The $\supseteq$ part.
The kernel of $F=\Fa$ contains the four vectors, because
\[
\begin{array}{l}
F(a_3, 0, -a_1) = ac + b0 + a_3(-a_1) = ac-ca = 0,\\
F(0, a_3, -a_2) = a0 + bc + a_3(-a_2) = bc-cb = 0,\\
F(a_2, -a_1, 0) = ab + a_2(-a_1) + c0 = ab-ba = 0.\\
\end{array}
\]
and
\[
F(1,1,1) = m(a_1 + a_2 + a_3) = 0.
\]

The $\subseteq$ part.
Let $K = \langle(a_3, 0, -a_1), (0,a_3,-a_2),(a_2, -a_1, 0), (1,1,1) \rangle$. $K$ is a
subgroup of $\Z^3$. By showing that the index $[\Z^3:K]$ is exactly the size $a_1+a_2+a_3$ of the image of $F$, we conclude that $K = \KerF$.
The subgroup $K$ is generated by the lines of the matrix
\[
M =
\left(\begin{array}{rrr}
1 & 1 & 1 \\
a_3 & 0 & -a_1 \\
0 & a_3 & -a_2 \\
a_2 & -a_1 & 0
\end{array}\right).
\]
From Lemma~\ref{lem:index}, the index is equal to the $\gcd$ of
the four $3$-minors of the matrix $M$:
\[
\begin{array}{ll}
\det
\left(\begin{array}{rrr}
1 & 1 & 1 \\
a_3 & 0 & -a_1 \\
0 & a_3 & -a_2 \\
\end{array}\right) = a_3(a_1+a_2+a_3), &
\det
\left(\begin{array}{rrr}
1 & 1 & 1 \\
a_3 & 0 & -a_1 \\
a_2 & -a_1 & 0
\end{array}\right) = -a_1(a_1+a_2+a_3)\\
\det
\left(\begin{array}{rrr}
1 & 1 & 1 \\
0 & a_3 & -a_2 \\
a_2 & -a_1 & 0
\end{array}\right) = -a_2(a_1+a_2+a_3), &
\det
\left(\begin{array}{rrr}
a_3 & 0 & -a_1 \\
0 & a_3 & -a_2 \\
a_2 & -a_1 & 0
\end{array}\right) = 0.
\end{array}
\]
That is the index is
\[
[\Z^3 : K] = \gcd(a_3(a_1+a_2+a_3),-a_1(a_1+a_2+a_3),-a_2(a_1+a_2+a_3),0) = a_1+a_2+a_3.\qedhere
\]
\end{proof}

\begin{corollary}
The kernel of $\Fa'$ is spanned by the vectors $a_3h_1-a_1h_3, a_3h_2-a_2h_3,a_2h_1-a_1h_2$.
\end{corollary}

\begin{proof}
This is because $\pi(\KerF)=\KerF'$. Indeed, $\Fa=\Fa'\circ\pi$ and $\pi(1,1,1)=0$.
\end{proof}

It is likely that the proposition and its corollary have an evident extension to any dimension. The proof requires a higher minor calculation. We leave this to the interested reader.

\subsection{The graph $\Ga$}

% Consider the set of vertices $\Z/s\Z$ and the set of edges $\bE''_d=(x,x+a_i)
% \mid x\in \Z/s\Z, 1\leq i \leq d\}$, where $a_i$ is taken modulo $s$. 

%\begin{definition}[Christoffel graph]
Let $d\geq2$ be an integer and $\a$ as before.
The graph $\Ga$ of normal vector $\a\in\Z^d$ is the directed graph 
$\Ga=\Fa(\Ha)$. It is also equal to
\[
    \Ga=\{(k,k+a_i) \mid k\in \Z/s\Z, 1\leq i\leq d\text{ and }  k < k+a_i\}.
\]
Two examples are shown at Figure~\ref{fig:christoffelgraphflat}.
\begin{figure}[h!]
\begin{center}
    \includegraphics{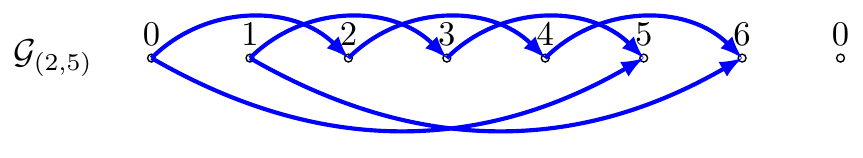}
    \includegraphics{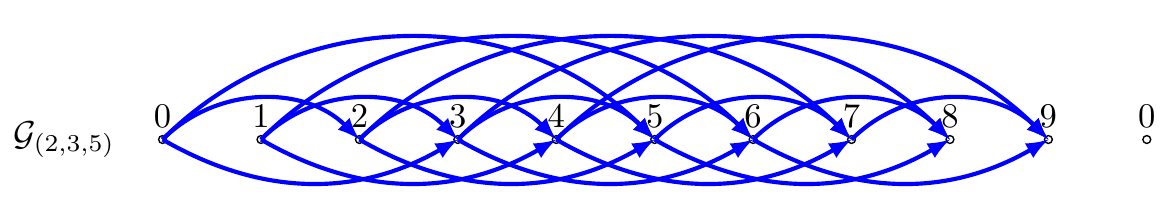}
\end{center}
\caption{The Christoffel graphs $\Ga$ for $\a=(2,5)$ and $\a=(2,3,5)$.}
\label{fig:christoffelgraphflat}
\end{figure}
Since the graph $\Ga$ is isomorphic to the quotients $\Ha / \KerF$ (and also
$\Ia/\KerF'$), we call it \emph{Christoffel graph} as well, because $\Ga$ is
the part of $\Ha$ in its fundamental domain.
The graph $\Ga$ can be embedded in a torus. Indeed, the graph $\Ia$ lives in
the diagonal plane $\D\simeq\R^{d-1}$ and is invariant under the group
$\KerF'$. The quotient $\D/\KerF'$ is a torus and contains the graph
$\Ia/\KerF'\simeq\Ga$ (see Figure~\ref{fig:christoffelgraphtorus}).

\begin{figure}[h!]
\begin{center}
\begin{tabular}{cccc}
$\G_{(2,5)}$ &
    \includegraphics{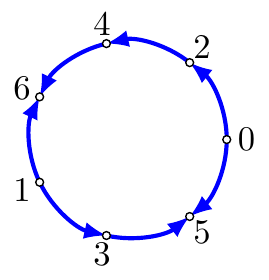}&
$\G_{(2,3,5)}$ &
    \includegraphics[width=0.21\linewidth]{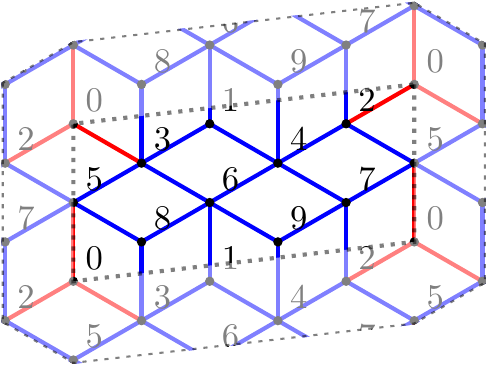}
\end{tabular}
\end{center}
\caption{The Christoffel graphs $\Ga$ for $\a=(2,5)$ and $\a=(2,3,5)$ can be
embedded in the torus $\D /\KerF'$.}
\label{fig:christoffelgraphtorus}
\end{figure}

\begin{figure}[h!]
\begin{center}
\begin{tabular}{ccc}
\includegraphics[width=0.21\linewidth]{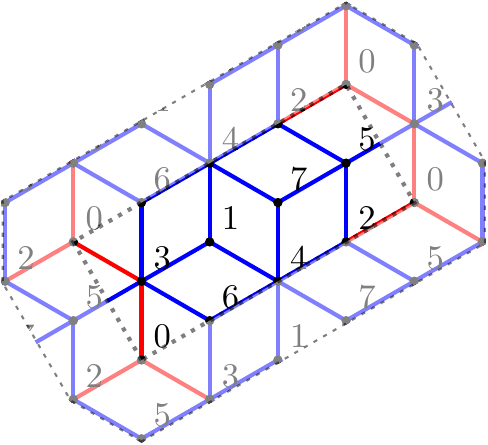} &
\includegraphics[width=0.21\linewidth]{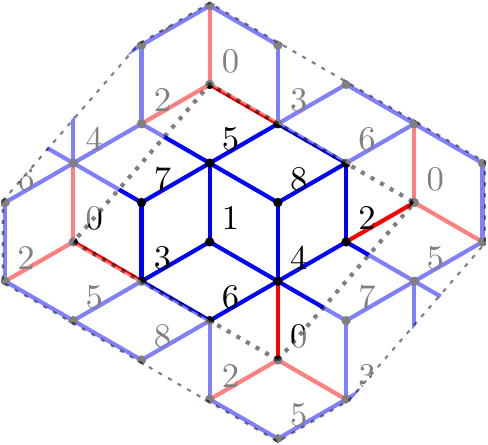} &
\includegraphics[width=0.21\linewidth]{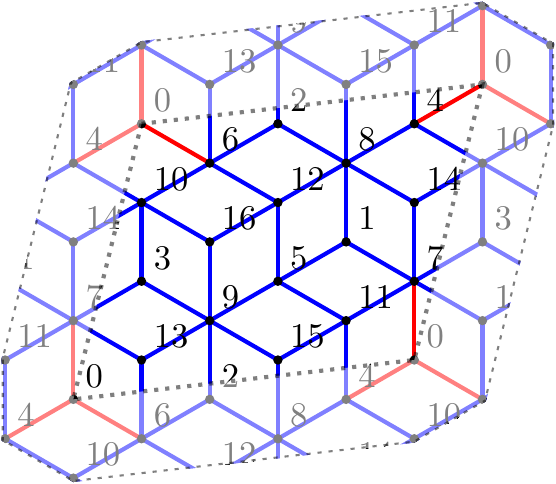}\\
$\G_{(2,3,3)}$ &
$\G_{(2,3,4)}$ &
$\G_{(4,6,7)}$
%\includegraphics{christo_graph_2_5.pdf}\\
%$\G_{(2,5)}$\\[-3.5cm]
\end{tabular}
\begin{tabular}{ccc}
\includegraphics[width=0.30\linewidth]{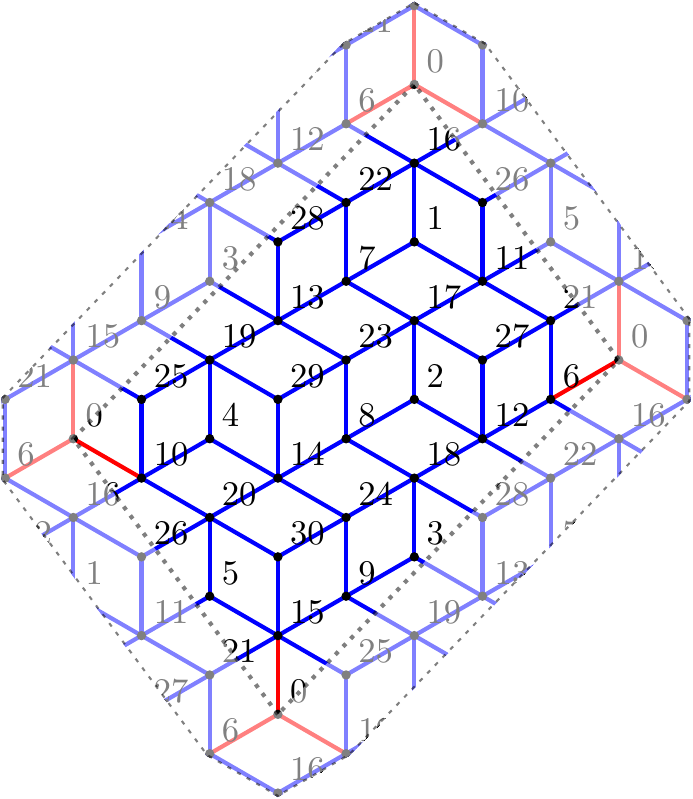} &
\includegraphics[width=0.36\linewidth]{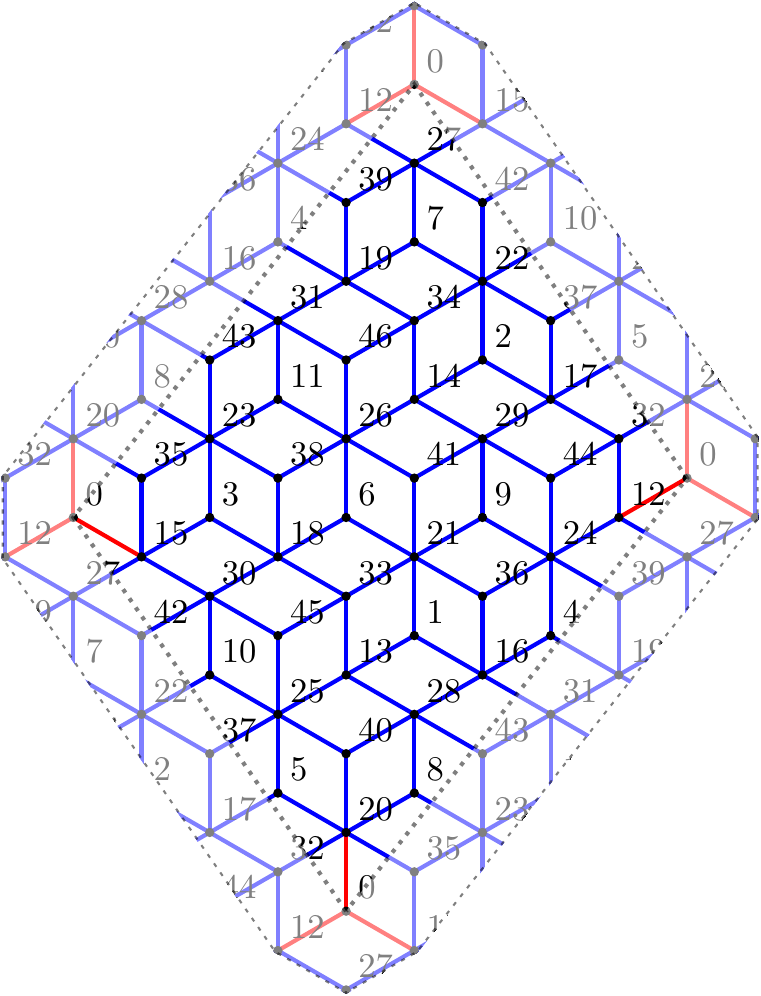} \\
%$\G_{(2,3,4,5)}$ &
$\G_{(6,10,15)}$ &
%$\G_{(7,11,13)}$ &
$\G_{(12,15,20)}$ \\
\end{tabular}
\end{center}
\caption{Some Christoffel graphs in dimension $d=3$. The body is blue, legs are in
red.}
\label{fig:cayleygraphs}
\end{figure}

The vertices of the Christoffel graph $\Ga$ and their image under the function
$\Fa$ corresponds to what is called \emph{roundwalk} in \cite{MR2074953}.
Their contribution allows to construct larger and larger domain of roundwalks
by iteration of extension rules.

The Christoffel graph has a natural representation inside $\Ia$. We define
this for $d=3$, leaving the generalizations for elsewhere. Recall that the
lattice $L$, defined in Section \ref{sec:IH}, is a free abelian group of rank 2,
spanned by the 3 vectors $h_1,h_2,h_3$ with $h_1+h_2+h_3=0$. Moreover, the
homomorphism $\Fa':L\to \Z/s\Z$ maps $h_i$ onto $a_i$, with $s=a_1+a_2+a_3$, and the $a_i$ are relatively prime; therefore, the mapping is surjective.

Choose some parallelogram in the plane $\D$ which is a fundamental domain
for its discrete subgroup $\KerF'$. We may assume that $O$ is a vertex of this parallelogram. Then $\Fa'$ induces a bijection 
between $\Z/s\Z$ and the integral points inside the parallelogram, excluding those lying on the two edges not containing $O$. It is such 
a parallelogram, with the part of the edges of $\Ia$ which lie inside him, that
we may call a {\em  Christoffel parallelogram}. This we may consider as the
generalization in dimension 3 of Christoffel words. Such a parallelogram tiles 
the plane $\D$ and completely codes the graph $\Ia$. Furthermore it is in
bijection with the Christoffel graph, as is easily 
verified. Examples are seen in Figure \ref{fig:cayleygraphs}.
%, where are also shown a generalization in dimension 4.

\begin{remark}\label{rem:compatible}
The graphs $\Ha,\Ia,\Ga$ are compatible, in the sense that $\Ia$ is the image
under $\pi$ of $\Ha$, $\Ga$ is the image under $F_\a$ of $\Ha$ and also the
image of $\Ia$ under $F'_\a$.
\end{remark}
\begin{center}
\begin{tikzpicture}[auto]
\node (A) at (0,0) {$\Ha$};
\node (B) at (6,0) {$\Ga$};
\node (C) at (3,-2) {$\Ia$};
\draw (A) edge[->] node {$\Fa$} (B);
\draw (C) edge[->] node {$\Fa'$} (B);
\draw (A) edge[->] node {$\pi$} (C);
\end{tikzpicture}
\end{center}

% \begin{proof}
% We have $L=\pi(\Z^d)$, $\Z/s\Z=\Fa(\Z^d)=\Fa'(L)$; this proves the result for the vertices. Regarding the edges, let $(u,v)$ be an edge of $\Ha$; then $v=u+e_i$ and $\Fa(u)<\Fa(v)$; thus $\Fa(u,v)$ is an edge of $\Ga$; conversely, an edge of $\Ga$ is, by definition of this graph, the image under $\Fa$ of an edge in $\Ha$. The other verifications are similar.
% \end{proof}

\subsection{The graph $\Hao$}

In this section, we extend the definition of Christoffel graphs to 
discrete plane such that the width $\omega$ is smaller than
$s=\Vert\a\Vert_1=\sum a_i$
where $\a\in \N^d$ is a vector of relatively prime positive integers as
before.
We consider only width $\omega$ such that $s/\omega$ is a positive integer
strictly smaller than the dimension $d$: $0<s/\omega<d$.
We define the mapping $\Fao:\Z^d\to \Z/\omega\Z$ sending each integral vector
$(x_1,\ldots,x_d)$ onto $\sum_ia_ix_i \bmod \omega$. 
We identify $\Z/\omega\Z$ and $\{0,1,\cdots, \omega-1\}$. A total order on
$\Z/\omega\Z$
is defined correspondingly.
The \emph{Christoffel graph of normal vector $\a\in\N^d$ of width $\omega$} is the
subset of edges $\Hao\subseteq\bE_d$ defined by
\[
%\Hao = \{ (u, u+e_i) \in \bE_d \mid \Fao(u) < \Fao(u+e_i) \}.
\Hao = \{ (u, v) \in \bE_d \mid \Fao(u) < \Fao(v) \}.
\]
This graph is related but does not correspond exactly to discrete plane of width
$\omega$. In fact, $\Hao$ can be obtained by the superposition of $s/\omega$
discrete plane of width $\omega$. The definition of $\Hao$ is motivated by
Pirillo's theorem, because this is what allows to generalize Pirillo's theorem
in arbitrary dimension (see Theorem~\ref{thm:pirilloforalld}).
Of course if $\omega=s$, then $\Hao=\Ha$ is the Christoffel graph of normal
vector $\a$. Also, if $d=2$ then $s=\omega$. 
If $d=3$, then either $\omega=s$ or $\omega=s/2$.
If $d=4$, then either $\omega=s$, $\omega=s/2$ or $\omega=s/3$ and so on for
$d\geq5$.
If $s$ is a prime number, then $\omega=s$.

As earlier, we define the projected graphs $\Gao:=\Fao(\Hao)$ and $\Iao:=\pi(\Hao)$.
The Christoffel graph $\Gao$ for the vector $\a=(15,11,10)$ of width $\omega=s=36$ is
shown at Figure~\ref{fig:H15_11_10_mod36and18} (left).
The Christoffel graph $\Gao$ for the vector $\a=(15,11,10)$ of width $\omega=18=s/2$ is
shown at Figure~\ref{fig:H15_11_10_mod36and18} (right) and a larger part is
shown at Figure~\ref{fig:H15_11_10_mod18}.
\begin{figure}[h!]
\begin{center}
\begin{tabular}{cc}
\includegraphics{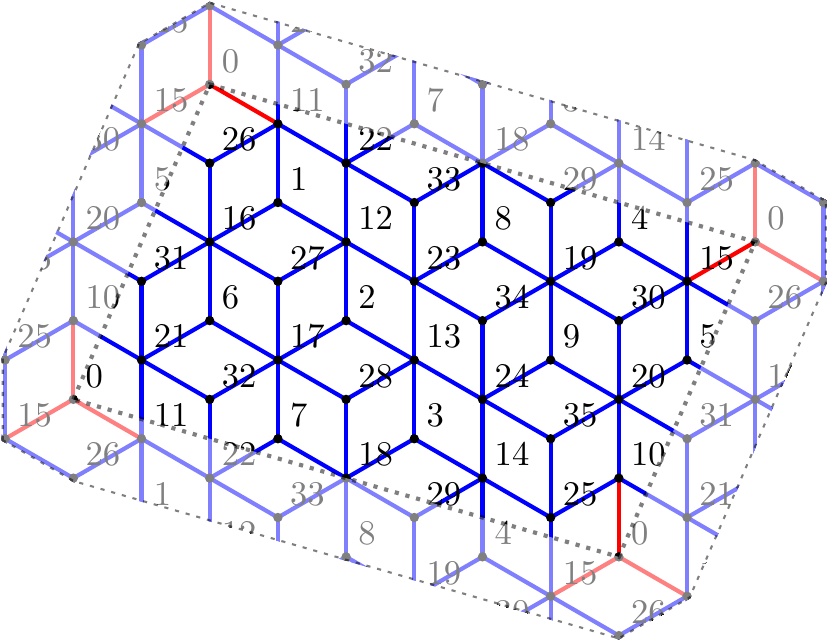} &
\includegraphics{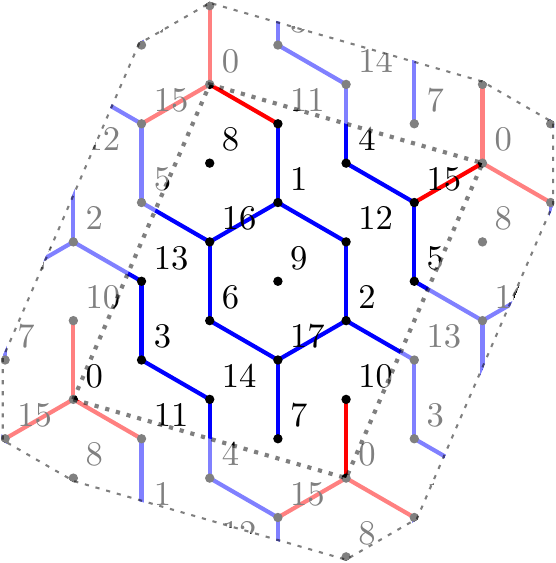}\\
$\G_{(15,11,10)} = \G_{(15,11,10),36}$ &
$\G_{(15,11,10),18}$
\end{tabular}
\end{center}
\caption{Left: the Christoffel graph $\Ga$ for the vector $\a=(15,11,10)$. 
Right: the Christoffel graph $\Gao$ of width $\omega=18$ for the vector
$\a=(15,11,10)$.}
\label{fig:H15_11_10_mod36and18}
\end{figure}

\begin{figure}[h!]
\begin{center}
\begin{tabular}{c}
\includegraphics{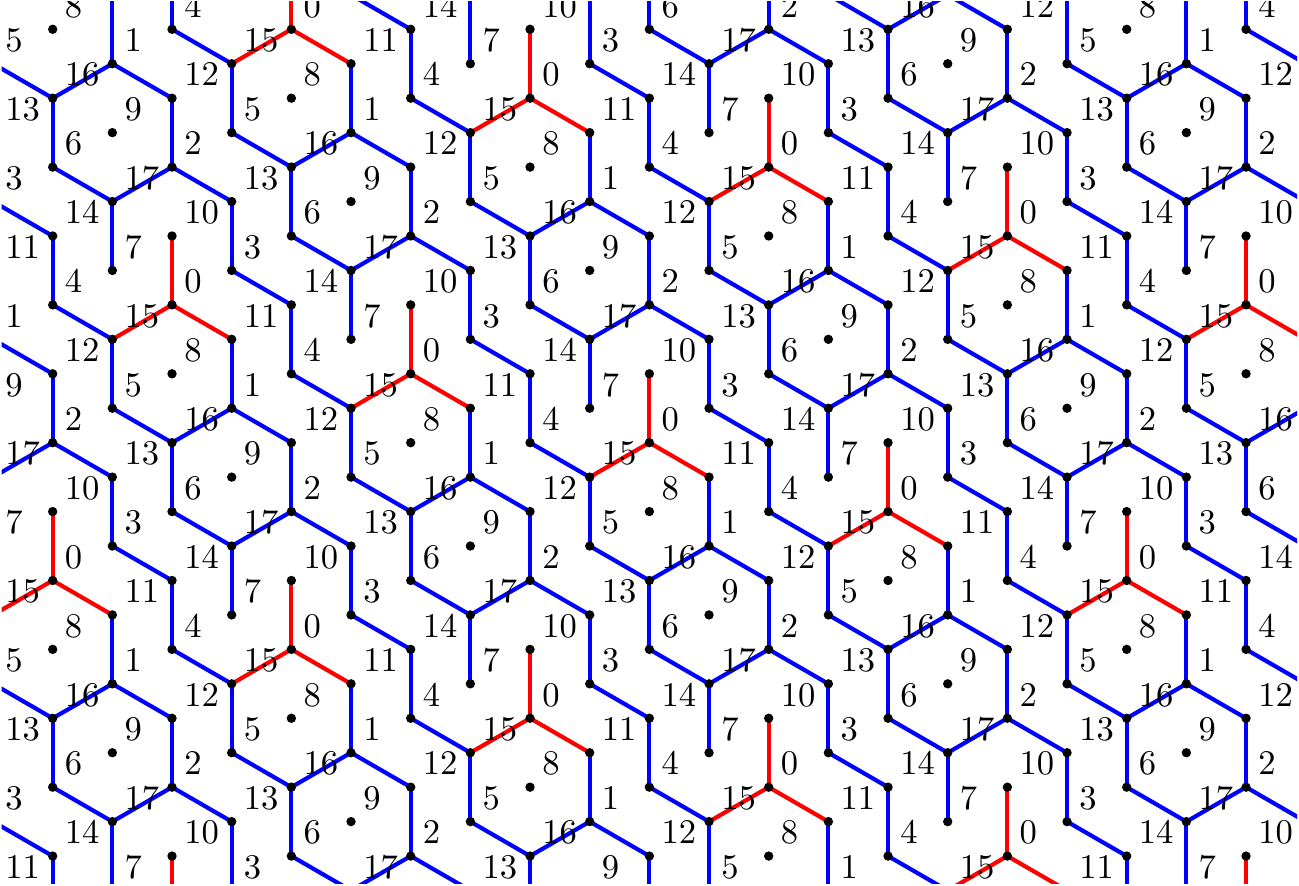}\\
$I_{(15,11,10),18}$
\end{tabular}
\end{center}
\caption{The Christoffel graph $\Iao$ of width $\omega=18$ for the vector
$\a=(15,11,10)$. It corresponds to the union of two
discrete planes of width $\omega$.}
\label{fig:H15_11_10_mod18}
\end{figure}

%%%%%%%%%%%%%%%%%%%%%
% ON VOIT RIEN EN 4D.
% See Figure~\ref{fig:christo2345mod7} for the Christoffel graph $\Hao$ of
% normal vector $\a=(2,3,4,5)$ of width $\omega=7$.
% \begin{figure}
% \begin{center}
% \includegraphics{christoffel_2_3_4_5_mod7.pdf}
% \end{center}
% \caption{
% The Christoffel graph of vector $\a=(2,3,4,5)$ modulo $s/2=7$.}
% \label{fig:christo2345mod7}
% \end{figure}
%%%%%%%%%%%%%%%%%%%%%

The next lemma gives an equivalent definition of the edges of the graph $\Hao$.

\begin{lemma}\label{lem:edgesintervalforHan}
Let $(u,v)\in\bE_d$ such that $v-u=e_i$ for some $1\leq i\leq d$. Then,
\begin{align}
(u,v)\in    \Hao \iff \Fao(u)\in[0,\omega-a_i-1] \iff \Fao(v)\in[a_i,\omega-1],\\
(u,v)\notin \Hao \iff \Fao(u)\in[\omega-a_i,\omega-1] \iff \Fao(v)\in[0,a_i-1].
\end{align}
\end{lemma}

\section{Flip, reversal and translation}\label{sec:fliprevtrans}

In this short section, we define the flip, reversal and translate of set of
edges. We define the operations for set of edges $X\subseteq\bE_d$ but they
extend naturally to set of edges of the form $\pi(X)$ and $\Fa(X)$ (see
Definition~\ref{def:commute} below). In order to define the flip operation, we
need to define the edges incident to zero.

\begin{definition}[edges of $\bE_d$ incident to zero]
Let $d\geq2$ be an integer and $\a\in\Z^d$ be a vector of relatively prime positive integers.
The set of \emph{edges of $\bE_d$ incident to zero} is
\[
\Qzero = \{(u,v)\in\bE_d : \Fa(u)=0 \text{ or } \Fa(v)=0 \}.
% \left( \{(0,e_i): 1\leq i\leq d\}\cup
% \{(-e_i,0) : 1\leq i \leq d\}\right) + \KerF\subseteq\bE_d.
\]
% Similarly, the set of {\em edges of $\bE'_d$ incident to zero} is the image
% under $\pi$ of the latter set. Finally, the set of {\em edges incident to
% zero} of $\bE''_d$ is the set of pairs of the form $(0,a_i)$ and $(a_i,0)$,
% $i=1,\ldots,d$.  
\end{definition}

%We denote all these sets by $\Qzero$, somewhat abusively, but there will no ambiguity.

\begin{definition}[body, legs]
Let $X\subseteq\bE_d$.
The set $X\setminus \Qzero$ is the \emph{body} and the edges of
$X\cap\Qzero$ are the \emph{legs} of $X$.
\end{definition}

See Figure \ref{fig:cayleygraphs} where the legs of graphs $\Ga$ are
represented in red, and the body in black.
%Similar definitions apply for the graphs $\Ia$ and $\Ga$. 
The $\flip$ is an operation which generalizes the function
$amb\mapsto bma$ defined for Christoffel words.
While we define the flip on graphs, it can also be seen as a flip in a rhombus
tiling when $d=3$ \cite{MR2856174,MR2440650,MR2330996}.

% Fernique:
% \cite{MR2856174}
% %    AUTHOR = {Bodini, Olivier and Fernique, Thomas and Rao, Michael and
% %              R{\'e}mila, {\'E}ric},
% %    TITLE = {Distances on rhombus tilings},
% \cite{MR2440650}
% %    AUTHOR = {Bodini, Olivier and Fernique, Thomas and R{\'e}mila, {\'E}ric},
% %     TITLE = {A characterization of flip-accessibility for rhombus tilings
% %              of the whole plane},
% \cite{MR2330996}
% %    AUTHOR = {Arnoux, Pierre and Berth{\'e}, Val{\'e}rie and Fernique,
% %              Thomas and Jamet, Damien},
% %     TITLE = {Functional stepped surfaces, flips, and generalized
% %              substitutions},

\begin{definition}[$\flip$]
For a subset of edges $X\subseteq\bE_d$, %(resp. $\bE'_d,\bE''_d$),
we define the \emph{$\flip$} operation which exchanges edges incident to zero:
\[
\flip: X \mapsto (X\setminus \Qzero) \cup (\Qzero\setminus X).
\]
\end{definition}
We see that $\flip(X)$ exchanges the legs of $X$ and keeps the body of $X$ invariant.
If $(u,v)\in\bE_d$,
then the reversal edge $(-v,-u)\in\bE_d$ is also an edge of the hypercubic
lattice and similarly for
the translated edge
$(u+t,v+t)\in\bE_d$
for all $t\in\Z^d$.
The reversal and translate operations extend on subsets 
of edges as follows:
\begin{definition}[Reversal, Translate]
Let $X\subseteq\bE_d$ be a subset of edges.
We define the \emph{reversal} $-X$ of $X$
and the \emph{translate} $X + t$,
for some $t\in\Z^d$,
of $X$ as
\[
-X = \{(-v,-u) \mid (u,v)\in X\}
\quad\quad
\text{and}
\quad\quad
X+t = \{(u+t,v+t) \mid (u,v)\in X\}.
\]
\end{definition}

\begin{definition}[$\flip$, Reversal, Translate]
\label{def:commute}
Let $X\subseteq\bE_d$.
The flip, reversal and translate of set of edges of the form $\pi(X)$ and
$\Fa(X)$ are defined naturally by commutativity:
\[
\begin{array}{lll}
\flip(\pi(X)):=\pi(\flip(X)), & -(\pi(X)):=\pi(-X), & \pi(X)+\pi(t):=\pi(X+t),\\
\flip(\Fa(X)):=\Fa(\flip(X)), & -(\Fa(X)):=\Fa(-X), & \Fa(X)+\Fa(t):=\Fa(X+t).
\end{array}
\]
%We have $-f(X)=f(-X)$, $f(X)+f(t)=f(X+t)$ and $\flip(f(X))=f(\flip(X))$.
\end{definition}

Therefore, statements proven for $\Ha$ using flip, reversal and translate
operations are also true for $\Ia$ and $\Ga$.  For example,
the goal of next section is to show that $\Ha+t=\flip(\Ha)$ for some
$t\in\Z^d$. If such an equation is satisfied for $\Ha$, it is clear from
Definition~\ref{def:commute} that $\Ia+\pi(t)=\flip(\Ia)$ and
$\Ga+\Fa(t)=\flip(\Ga)$.

\section{Flipping is translating}\label{sec:flipistranslating}

In this section, we show that the flip of the Christoffel graph $\Ha$ is a translate
of $\Ha$; this  is a generalization of one implication of
Theorem~\ref{thm:pirillo}. We also show that the body of $\Ha$ is
symmetric and as a consequence we obtain that a Christoffel graph is a translate of
its reversal.
The results stated in this section are stated and proved for $\Ha$ but they
are valid for $\Ia$ and $\Ga$ by Definition~\ref{def:commute}.

The following lemma describes the legs of $\Ha$.
\begin{lemma}[Legs of $\Ha$]\label{lem:legs}
An edge $(u,v)$ is a leg of $\Ha$ if and only if $\Fa(u)=0$.
%\[
%\Ha\cap\Qzero =
%\{(u,v)\in\bE_d : \Fa(u)=0 \}.
%%\{(0,e_i): 1\leq i\leq d\} + \KerF\subseteq\bE_d.
%\]
\end{lemma}

\begin{proof}
We have $(-e_i,0)\notin\Ha$ because
there is no $u\in \Z^d$ such that $\Fa(u) < \Fa(0)= 0$.
Moreover $(0,e_i)\in E$ for each $i$, $1\leq i\leq d$, because
$\Fa(0)=0<a_i=\Fa(e_i)$. 
\end{proof}

We now show that the body of a Christoffel graph is {\em symmetric}, i.e., it is equal to its reversal.
This generalizes the fact that central words are palindromes.

\begin{lemma}\label{lem:bodysymmetric}
The body of $\Ha$ %(resp. of $\Ia, \Ga$) 
is symmetric, i.e., $-(\Ha\setminus\Qzero)  = \Ha\setminus\Qzero$.
\end{lemma}
\noindent

%In the following proof, we note $s=\sum_{i=1}^{d} a_i$.

\begin{proof}
It is sufficient to prove $-(\Ha\setminus \Qzero) \supseteq \Ha\setminus
\Qzero$, the other inclusion being equivalent, since symmetry is involutive.
Let $(u,v)\in \Ha\setminus \Qzero$. Then $v - u =e_i$ for some $1\leq i \leq
d$.
Then $\Fa(u)\in[0,s-a_i-1]$ by Lemma~\ref{lem:edgesinterval}
and $\Fa(u)\notin\{0,s-a_i\}$ so that
$\Fa(u)\in[1,s-a_i-1]$.
Thus $\Fa(-u)=s-\Fa(u)\in[a_i+1,s-1]$.
We obtain that $(-v,-u)\in\Ha$ by Lemma~\ref{lem:edgesinterval} because
$-u - (-v)=v - u = e_i$.
Since $\Qzero=-\Qzero$, $(u,v)\notin\Qzero$ implies that $(-v,-u)\notin\Qzero$.
We conclude $(-v,-u)\in \Ha\setminus \Qzero$
and $(u,v) \in -(\Ha\setminus \Qzero)$. 
\end{proof}

Now we show that the reversal is equal to the flip of a Christoffel graph.
This generalizes the fact that the reversal $\til{amb}$ of a Christoffel word
is equal to $bma$.

\begin{lemma}\label{lem:symmetry}
The reversal of $\Ha$ 
%(resp $\Ia, \Ha$)
is equal to its flip, i.e., $-\Ha = \flip(\Ha)$.
\end{lemma}
\begin{proof}
For $\Ha$, we have to show that $-\Ha = (\Ha\setminus \Qzero) \cup (\Qzero\setminus \Ha)$. We prove the result in two parts since $-\Ha = \left((-\Ha)\setminus
\Qzero\right) \cup
\left((-\Ha)\cap \Qzero\right)$ is the disjoint union of a part outside of $\Qzero$ and a
part inside of $\Qzero$.
Outside of $\Qzero$: since $\Qzero$ is symmetric and because $\Ha$ is
symmetric from Lemma~\ref{lem:bodysymmetric}, we have
$-(\Ha)\setminus \Qzero = -(\Ha)\setminus -\Qzero = -(\Ha\setminus \Qzero) =
\Ha\setminus \Qzero$.
Inside of~$\Qzero$:
we have $(-\Ha)\cap \Qzero = -(\Ha\cap \Qzero)
%=\{(-v,-u) \in\bE_d: \Fa(u)=0\}
=\{(u,v) \in\bE_d: \Fa(v)=0\}
= \Qzero\setminus\Ha$ by Lemma~\ref{lem:legs}.
\end{proof}

Now we show that the flip of a Christoffel graph $\Ha$ is equal to a translate of
$\Ha$.
It generalizes one implication of Theorem~\ref{thm:pirillo}.
It corresponds to the fact that a Christoffel word $amb$ is conjugate to
to $bma$. 
%This will be shown to be a characteristic property of Christoffel words.
Proposition~\ref{prop:translation} is illustrated in Figure~\ref{fig:H25flip}
and Figure~\ref{fig:I467flip}.

\begin{proposition}\label{prop:translation}
Let $t\in\Z^d$ %(resp. $\in L, \Z/s\Z$) 
be such that $\Fa(t)=1$. %(resp. $\Fa'(t)=1, t= 1 \mod s$).
The translate by $t$ of $\Ha$ %(resp. $\Ia, \Ga$) 
is equal to its flip, i.e.,
$\Ha+t = \flip(\Ha)$.
\end{proposition}

Note that, since $\Fa$ is surjective, there exists indeed $t\in\Z^d$ such that $\Fa(t)=1$.
Also, $\Fa(-t)=-\Fa(t)=-1$.

\begin{proof}
We prove the result in two parts since $\Ha+t$ is the disjoint union
of a part outside of $\Qzero$ and a part inside of $\Qzero$:
\[
\Ha+t = \left((\Ha+t)\setminus \Qzero\right)\cup \left((\Ha+t)\cap
\Qzero\right).
\]

1. $(\Ha+t)\setminus \Qzero \supseteq \Ha\setminus \Qzero$.
Suppose that $(u,v)\in \Ha\setminus\Qzero$. Thus
$v - u =e_i$ for some $1\leq i \leq d$.
Then, $\Fa(u)\in[0,s-a_i-1]$ by Lemma~\ref{lem:edgesinterval} and, since the edge is not a leg,
$\Fa(u)\notin\{0,s-a_i\}$.
Hence, $\Fa(u)\in[1,s-a_i-1]$. Then
$\Fa(u-t)=\Fa(u)-1\in[0,s-a_i-2]$ which implies that $(u-t,v-t)\in\Ha$.
Then $(u,v)\in\Ha+t$.

2. $(\Ha+t)\setminus \Qzero \subseteq \Ha\setminus \Qzero$.
Let $(u+t,v+t)\in(\Ha+t)\setminus\Qzero$ for some edge $(u,v)\in\Ha$.
Then, $\Fa(u+t)\notin\{0,s-a_i\}$.
From Lemma~\ref{lem:edgesinterval}, $\Fa(u+t)=\Fa(u)+1\in[1,s-a_i]$.
Therefore, $\Fa(u+t)\in[1,s-a_i-1]$ and we conclude that $(u+t,v+t)\in \Ha$.

3. $(\Ha+t)\cap \Qzero \supseteq \Qzero\setminus\Ha$.
Let $(u,v)\in\Qzero\setminus\Ha$.
Then, $\Fa(v)=0$ which implies that $\Fa(v-t)=s-1$.
By Lemma~\ref{lem:edgesinterval}, we have $(u-t,v-t)\in\Ha$ so that
 $(u,v)\in\Ha+t$.

4. $(\Ha+t)\cap \Qzero \subseteq \Qzero\setminus\Ha$.
Let $(u+t,v+t)\in(\Ha+t)\cap \Qzero$ for some edge $(u,v)\in\Ha$.
Either $\Fa(u+t)=0$ or $\Fa(v+t)=0$.
If $\Fa(u+t)=0$, then $\Fa(u)=s-1$ which implies that $(u,v)\notin\Ha$ by
Lemma~\ref{lem:edgesinterval}, a contradiction.
One must have $\Fa(v+t)=0$, which implies that $(u+t,v+t)\notin\Ha$.
\end{proof}

\begin{figure}[h]
\begin{center}
\includegraphics{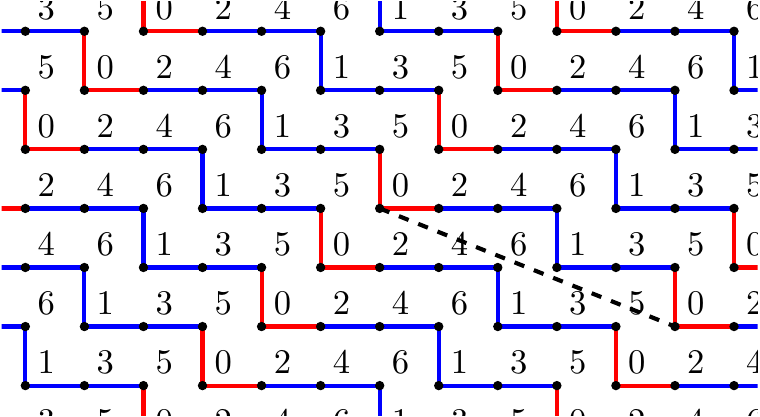}
\quad\quad
\includegraphics{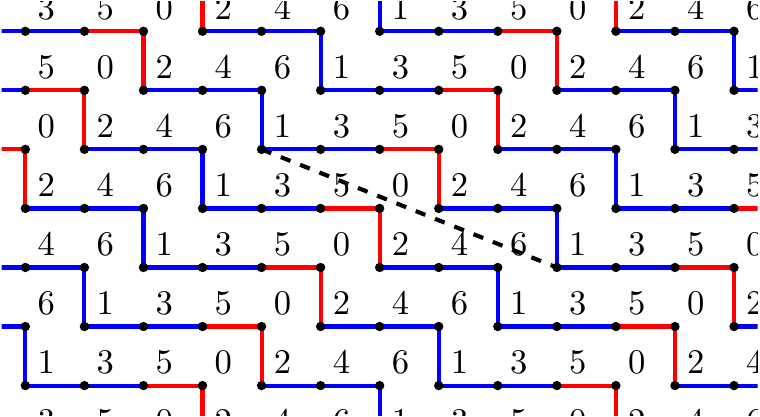}
\end{center}
\caption{Left: the graph $\Ha$ with $\a=(2,5)$.
Right: $\flip(\Ha)$.}
\label{fig:H25flip}
\end{figure}

\begin{figure}[h!]
\begin{center}
\includegraphics[width=.47\linewidth]{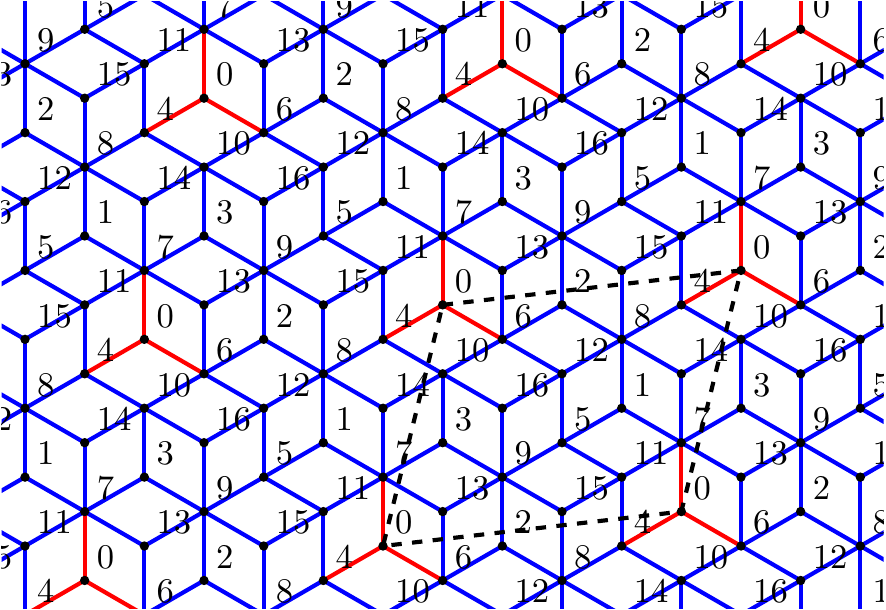}
\quad\quad
\includegraphics[width=.47\linewidth]{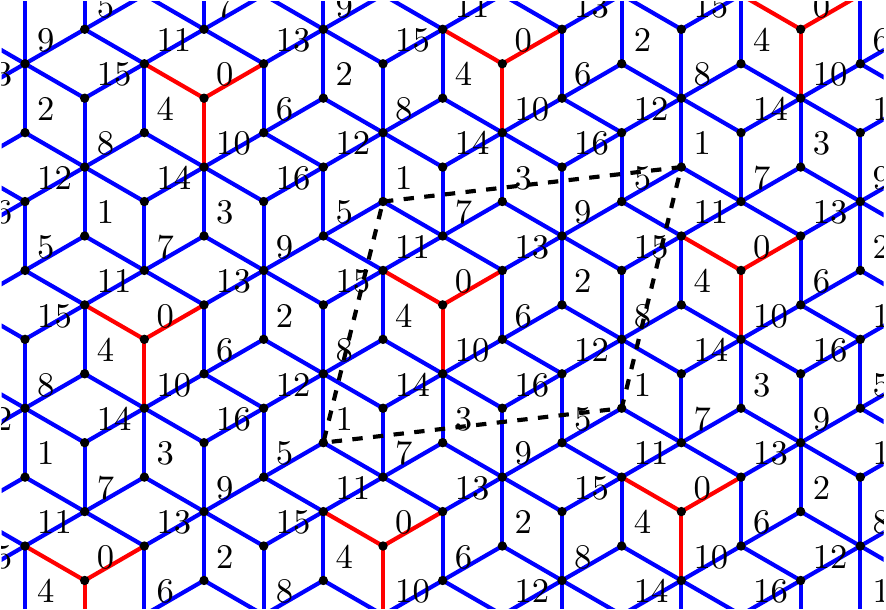}
\end{center}
\caption{Left: the graph $\Ia$ with $\a=(4,6,7)$.
Right: $\flip(\Ia)$.
Consider the Christoffel parallelogram $P$ with vertices labeled by $0$
embedded in $\Ia$. 
The parallelogram $P$ also appears in the graph $\flip(\Ia)$
with vertices labeled by $1$.
% The parallelogram obtained by symmetry of $P$ with respect
% to its center appears as a translate of $P$ within $\Ia$: it corresponds to
% the parallelogram with vertices labeled by $-1=9$. It is also equal to the flip of $P$.
}
\label{fig:I467flip}
\end{figure}

The previous proposition proves that the body of a Christoffel graph $\Ha$
has a period. This generalizes the fact that central words of length $p+q-2$
have periods $p$ and $q$ (remark that $p$ and $-q$ is the same period mod
the length of the Christoffel words $|w|=p+q$).
Indeed, let $P$ be some parallelogram and $M$ be an
inner point. Consider the $4$ vectors with origin equal to one of the vertices
of $P$ and with end $M$.  Then, for each point $X$ in $P$, there is one of
these vector, $\vec v$ say, such that the segment $[X,X+\vec v]$ is contained
in $P$. We leave the verification of this to the reader. It follows that a
Christoffel parallelogram may be reconstructed from the edges
incident to zero by applying translations which stay completely in the
parallelogram. This is completely analoguous to the fact that a central word
is completely determined by its two periods.

%The next result states that the Christoffel graph $\Ha$ of normal vector $\a$
%is a translate of its reversal.
The next result generalizes the fact that the reversal $\til{w}$ of a Christoffel word $w$
is conjugate to $w$. This is not a characteristic property of Christoffel
words, because it is satisfied for all words that are
the product of two palindromes.

\begin{corollary}\label{cor:conjugatetoreversal}
Let $t\in\Z^d$ be such that $\Fa(t)=1$.
Then $-\Ha = \Ha + t$. %, $-\Ia = \Ia + \pi(t)$, $-\Ga = \Ga + \Fa(t)$.
\end{corollary}

\begin{proof}
Follows from Lemma~\ref{lem:symmetry} and Proposition~\ref{prop:translation}.
\end{proof}

\begin{corollary}\label{ChristParal}
(i) The body of a Christoffel parallelogram is symmetric with respect to its center.

(ii) Consider a Christoffel parallelogram $P$ embedded in $\Ia$. The
parallelogram obtained by symmetry of $P$ with respect to its center appears
as a translate of $P$ within $\Ia$, and is also equal to the flip of $P$.
\end{corollary}

We thus have obtained a generalization of: (i) a central word is a palindrome;
(ii) the reversal of a Christoffel word $amb$ is conjugate to it, and equal to
$bma$. The corollary can be checked on Figures~\ref{fig:cayleygraphs},
\ref{fig:H25flip} and \ref{fig:I467flip}

\section{Higher-dimensional Pirillo's theorem}\label{sec:Ddimpirillothm}

In this section, we study the converse of
Proposition~\ref{prop:translation}. In other words, does the fact of being a
translate of its flip is a characteristic property of Christoffel graphs as it
is the case for Christoffel words?
We show that it must a Christoffel graph $\Hao$ for some vector $\a\in\Z^d$
and width $\omega$.
If $d=3$, we show that parallelograms that are translate to their flip are
Christoffel parallelograms or their edge-complement.

Let $K$ be a subgroup of $\Z^d$ for some integer $d\geq2$ such that the
index $[\Z^d:K]$ is finite and $\sum_{i=1}^d e_i = (1,1,\ldots, 1) \in
K$.
Let $\Qzero$ be the set of \emph{edges of $\bE_d$ incident to zero mod $K$}:
\[
\Qzero = \{(u,v)\in\bE_d \mid u\in K \text{ or } v\in K \}.
%\Qzero = \left( \{(0,e_i)\mid 1\leq i\leq d\}\cup
%\{(-e_i,0) \mid 1\leq i \leq d\}\right) + K\subseteq\bE_d.
\]
For a subset of edges $X\subseteq\bE_d$, we redefine the $\flip$ operation
according to the above set $\Qzero$:
\[
\flip: X \mapsto (X\setminus \Qzero) \cup (\Qzero\setminus X).
\]
In what follows, we assume that $M\subseteq \bE_d$ is a set of
edges such that
\begin{itemize}
\item $M$ is invariant for the group of translations $K$;
\item $\flip(M)=M+t$ for some $t\in\Z^d$. 
\end{itemize}
If $\flip(M)=M+t$, then for each $i$, $(0,e_i)\in M$ or $(-e_i,0)\in M$ but
not both. Otherwise the number of edges parallel to the vector $e_i$ is not
preserved by the $\flip$ and the equation can not be satisfied.
Therefore, we suppose that for each $i$, $1\leq i\leq d$, $(0,e_i)\in M$ and
$(-e_i,0)\notin M$. In
other words, the legs of $M$ are:
\begin{itemize}
%\item $\Qzero \cap M=\{(0,e_1), (0,e_2), \ldots, (0,e_d)\}+K$.
\item $\Qzero \cap M=\{(u,u+e_i)\in\bE_d \mid u\in K \}$.
\end{itemize}
The question we consider in this section is:
for which set of edges $M\subseteq\bE_d$ satisfying the above three conditions
does there exist a translation $t\in\Z^d$ such that $\flip(M)=M+t$ (see
Figure~\ref{fig:motifM}).
\begin{figure}[h]
\begin{center}
\includegraphics{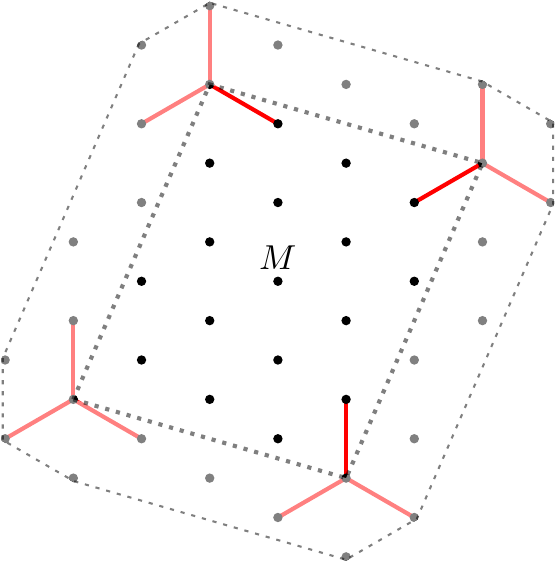}
\quad\quad
\includegraphics{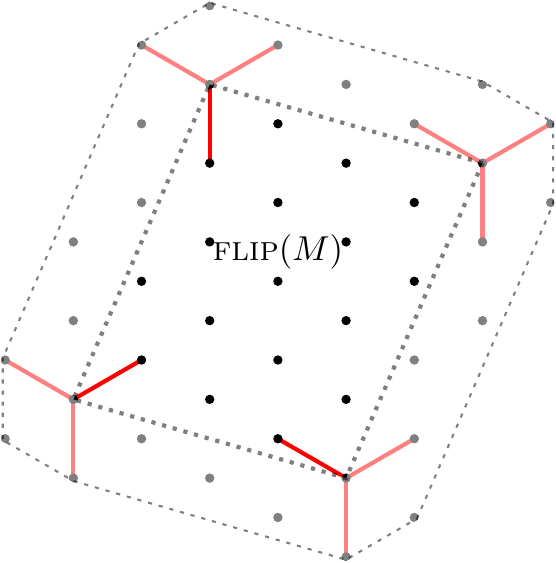}
\end{center}
\caption{Left: $M$.
    Right: $\flip(M)$ for the subgroup
    $K=\langle(0,4,1),(-2,0,3),(1,1,1)\rangle$.
}
\label{fig:motifM}
\end{figure}

%The next lemma is used in
%the proposition that follows.

\begin{lemma}\label{lem:ST}
Let $t\in\Z^d$ be a translation.
Let $X\subseteq\bE_d$ and $h\in\bE_d$. We have
\begin{enumerate}[\rm(i)]
  \item If $h\in X$, then $h+t\notin \Qzero$ if and only if $h+t\in \flip(X+t)$.
  \item If $h\notin X$, then $h+t\in \Qzero$ if and only if $h+t\in \flip(X+t)$.
\end{enumerate}
\end{lemma}
\begin{proof}
We have $\flip(X+t)=((X+t) \setminus \Qzero) \cup (\Qzero \setminus
(X+t))$.

(i) If $h\in X$ then $h+t\in X+t$.
If $h+t\notin \Qzero$,
then $h+t\in (X+t)\setminus\Qzero\subseteq \flip(X+t)$.
If $h+t\in \Qzero$,
then $h+t\in (X+t)\cap\Qzero$.
Therefore $h+t\notin\flip(X+t)$.

(ii) If $h\notin X$ then $h+t\notin X+t$.
If $h+t\in \Qzero$,
then $h+t\in \Qzero\setminus(X+t)\subseteq\flip(X+t)$.
If $h+t\notin \Qzero$,
then $h+t\notin (X+t)\cup\Qzero\supseteq \flip(X+t)$.
Therefore $h+t\notin\flip(X+t)$.
\end{proof}

%\todo{add image for $X$ and $\flip(X)$}

%\todo{add image for $M$, $M+t$ and $\flip(M+t)$}

\begin{proposition}\label{prop:existbi}
%If there exists $t\in\Z^d$ such that $M=\flip(M+t)$, then 
For all $i$, with $1\leq i\leq d$, there exists
a unique integer $b_i$, $0<b_i<\omega$, such that $e_i+b_i t\in K$ where
$\omega$ is the order of
$t$ in the group $\Z^d/K$. Moreover,
\[
(0, e_i) + kt
\begin{cases}
\in M & \text{if}\quad 0 \leq k < b_i,\\
\notin M & \text{if}\quad b_i \leq k < \omega,
\end{cases}
\]
\end{proposition}

%\begin{figure}[h!]
%\begin{center}
%\includegraphics{tikz_ei_ai.pdf}
%\end{center}
%\caption{Edges incident to zero that are in $M$. An edge $h$ is drawn if and only if $h\in
%\Qzero\cap M$.}
%\label{fig:ei_ai}
%\end{figure}

In the following proof, for two elements $u,u'\in\Z^d$ the notation $u\equiv
u'$ is used when $u'-u\in K$. The notation is also used for two edges
$(u,v),(u',v')\in\bE_d$: $(u,v)\equiv(u',v')$ if and only if $u'-u=v'-v\in K$.

\begin{proof}
Let $\omega$ be the order of $t$ in $\Z^d/K$. In this proof, we denote by $\vec{0}$ the zero of $\Z^d/K$. Thus let
\[
\omega = \min\{k>0 | kt \in K \} = {\rm order}_{\Z^d/K}(t)
\]
Consider the orbit under the translation $t$ of the edge $h=(\vec{0},e_i)\in
\Qzero\cap M$. We have that $h+\omega t\equiv h\in \Qzero\cap M$. We want to show that there exists
$b_i$, such that $0<b_i<\omega$ and $h+b_it \equiv (-e_i,\vec{0}) \in \Qzero$.

Suppose (by contradiction) that $h+kt\notin \Qzero$ for all $0<k<\omega$.
From Lemma~\ref{lem:ST} (i),
$h\in M$ and $h+t\notin\Qzero$, then $h+t\in\flip(M+t)=M$.
Recursively, we have
$h+kt\in \flip(M+t)=M$ for all $0<k<\omega$. This is summarized in the following
graph:
\[
\begin{array}{cccccccccccc}
h     & \xrightarrow{+t} &
h+t  & \xrightarrow{+t} &
h+2t  & \xrightarrow{+t} &
\cdots  & \xrightarrow{+t} &
h+(\omega-1)t  & \xrightarrow{+t} &
h+\omega t\equiv h  \\
\\
\in \Qzero\cap M & &
\in M\setminus \Qzero  & &
\in M\setminus \Qzero  & &
  & &
\in M\setminus \Qzero  & &
\in \Qzero
%\in \Qzero\setminus M
\end{array}
\]
But then, $h+(\omega-1)t\in M$ and $h+\omega t\in\Qzero$, so that $h+\omega t\notin\flip(M+t)=M$ from
Lemma~\ref{lem:ST}~(i).
This is a contradiction because $h+\omega t\equiv h\in M$. Hence, there
must exist some $b_i$, $0<b_i<\omega$ such that $h+b_it\in \Qzero$. Since $h$ is an
edge parallel to the vector $e_i$, then
either $h+b_it\equiv (\vec{0},e_i)$ or $h+b_it\equiv (-e_i,\vec{0})$. The first option
contradicts the minimality of $\omega$. We conclude that $e_i+b_it\equiv \vec{0}$.

The number $b_i$ is also unique. Indeed, suppose there exist $b_i$ and $b_i'$
with $0<b_i<b_i'<\omega$ such that $h+b_it\equiv h+b_i't\equiv (-e_i,\vec{0})$. Then
$(b_i'-b_i)t = (\vec{0}+b_i't)-(\vec{0}+b_it)\equiv(-e_i)-(-e_i) = \vec{0}$.
This contradicts the minimality of $\omega$ since $0<b_i'-b_i<\omega$.

% Let $h=(0,e_i)$.
% From Proposition~\ref{prop:existbi},
From the above paragraph,
we have that $h+kt\notin \Qzero$ for all $k$ such
that $0<k<b_i$ or $b_i<k<\omega$.
Using Lemma~\ref{lem:ST}~(i), if $h+(k-1)t\in M$ and
$h+kt\notin\Qzero$, then $h+kt\in\flip(M+t)=M$.
Thus by recursion $h+kt\in M$ for all $k$ with $0<k<b_i$.
Also $h+b_it\equiv(-e_i,\vec{0})\in\Qzero\setminus M$.
Using Lemma~\ref{lem:ST}~(ii), if $h+b_it\notin M$ and
$h+(b_i+1)t\notin\Qzero$, then $h+(b_i+1)t\notin\flip(M+t)=M$.
Thus by recursion $h+kt\notin M$ for all $k$ with $b_i<k<\omega$.
\end{proof}

\begin{lemma}\label{lem:generatedbyt}
$\Z^d/K$ is cyclic and generated by $t$.
\end{lemma}

\begin{proof}
Let $u=(x_1,x_2,\cdots,x_d)\in\Z^d$.
Using Proposition~\ref{prop:existbi}, we have
\[
u = \sum x_i e_i \equiv \sum x_i (-b_i t) = - \sum(b_ix_i) t
\]
Let $k=-\sum (b_ix_i)\mod \omega$. Then, $0\leq k < \omega$ and $u=kt$.
\end{proof}

\begin{lemma}\label{lem:MgeneratedbyST}

\[
M = \{(0,e_i) + kt : 1\leq i\leq d \text{ and } 0\leq k < b_i\} + K.
\]
\end{lemma}

\begin{proof}
($\supseteq$) If $0\leq k<b_i$, then $(0,e_i)+kt\in M$ by
Proposition~\ref{prop:existbi}.

($\subseteq$) Let $(u,u+e_i)\in M$ with $u\in\Z^d$.
From Lemma~\ref{lem:generatedbyt}, $(u,u+e_i)=(0,e_i)+u \equiv (0,e_i)+kt$ for
some integer $k$ such that $0\leq k < \omega$.
From Proposition~\ref{prop:existbi}, $0\leq k <b_i$.
\end{proof}

For all $i$ with $1\leq i\leq d$, let $a_i$ be such that $a_i+b_i=\omega$.
Also let
\[
\b=(b_1,b_2,\cdots,b_d)
\quad
\text{and}
\quad
\a=(a_1,a_2,\cdots,a_d)
\]
We have $a_it = (\omega-b_i)t=\omega t-b_it \equiv e_i$.
%(see Figure~\ref{fig:ei_ai}).
The next result shows that $\omega$ is a divisor of $\sum a_i$ and $\sum b_i$.

%\todo{what does it mean if $\sum a_i$ or $\sum b_i$ is prime?}

\begin{lemma}\label{lem:existsdivisor}
%If $M=\flip(M+t)$, then 
There exist integers $q$ and $\ell$, with $0<q<d$ and
$0<\ell<d$, such that $\omega\cdot q=a_1+a_2+\cdots+a_d$ and $\omega\cdot
\ell=b_1+b_2+\cdots+b_d$.
Moreover $d=q+\ell$.
\end{lemma}

\begin{proof}
For all $1\leq i\leq d$, we have $e_i = a_it = -b_it$. Thus,
$-(b_1+b_2+\cdots+b_d)t$ is an overall translation of $e_1+e_2+\cdots+e_d\in
K$, i.e., the identity.
Similarly, $(a_1+a_2+\cdots+a_d)t=e_1+e_2+\cdots+e_d\in K$.
Therefore, the order of $t$ ($=\omega$) must divide both $a_1+a_2+\cdots+a_d$ and
$b_1+b_2+\cdots+b_d$.
Then, there exist integers $q$ and $\ell$ such that $\omega\cdot
q=a_1+a_2+\cdots+a_d$ and $\omega\cdot \ell=b_1+b_2+\cdots+b_d$.
But $a_i<\omega$ for each $i$ so that $a_1+a_2+\cdots+a_d<d\omega $ and $q<d$. Similarly,
$\ell<d$.
Finally, $\omega q+\omega\ell=\sum (a_i+b_i)=\omega d$ and therefore $d=q+\ell$.
\end{proof}

If the sum of the $a_i$ or the sum of the $b_i$ is $\omega$, then the next two
theorems claim that $M$ is closely related to the Christoffel graph.

\begin{theorem}\label{thm:Mischristo}
(i) If $\sum a_i=\omega$, then $M=\Ha$;

(ii) if $\sum b_i=\omega$, then the complement $M^c=\bE_d\setminus M$ of $M$ is equal to 
$-H_\b$.
\end{theorem}

\begin{proof}
(i) For all $u=(x_1,x_2,\cdots,x_d)\in\Z^d$, we have $u=kt$ with
\[
k
= \sum (-b_ix_i)\bmod \omega
= \sum (a_i-\omega)x_i\bmod \omega
= \sum a_ix_i\bmod \sum a_i
= \Fa(u)
\]
We want to show that $M=\Ha$.
We have that $(u,u+e_i)=(0,e_i)+kt\in M$ if and only if $0\leq k<b_i$ if and
only if $\Fa(u)\in[0,\omega-a_i-1]$ if and only if $(u,u+e_i)\in\Ha$.

(ii) For all $u=(x_1,x_2,\cdots,x_d)\in\Z^d$, we have $u \equiv kt$ with
\[
k
= \sum (-b_ix_i)\bmod \omega
= -\sum b_ix_i\bmod \sum b_i
= -\Fb(u)
\]
We want to show that $M^c=-\Hb$.
We have that $(u,v)=(u,u+e_i)=(0,e_i)+kt\in\bE_d\setminus M$
if and only if $b_i\leq k<\omega$
%if and only if $b_i\leq -\Fb(u) < \omega$
if and only if $\Fb(-u)\in[b_i,\omega-1]$
if and only if $(-u-e_i,-u)\in\Hb$
if and only if $(u,v)\in-\Hb$.
\end{proof}

%\begin{proposition}[Pirillo]
%If $d=2$ and $M=\flip(M+t)$, then $M$ is a Christoffel graph $\Ga$ of normal
%vector $\a$. Also, $M$ is the complement of the reversal of the Christoffel
%graph $\Gb$ of normal vector $\b=(b_1,b_2)=(a_2,a_1)$.
%\end{proposition}
%
%\begin{proof}
%From Lemma~\ref{lem:existsdivisor} there exist an integers $0<q<2$ and $0<\ell<2$ such that
%$nq=a_1+a_2$ and $\omega\ell=b_1+b_2$. Of course $q=\ell=1$ is the only possibility
%and thus $\omega=b_1+b_2=a_1+a_2$.
%Also $a_1+b_1=a_2+b_2=\omega$ so that $a_1=b_2$ and $a_2=b_1$.
%The hypothesis of both Theorems
%\ref{thm:reversalcomplementofMischristo} and \ref{thm:Mischristo} are
%satisfied.
% We conclude that $M$ is a Christoffel graph $\Ga$ for vector
%$\a=(a_1,a_2)$.  Moreover, the complement of $M$ is also
%the reversal of a Christoffel graph, namely, $M^c=-\Gb$ for
%$\b=(b_1,b_2)=(a_2,a_1)$.
%\end{proof}
%
%For example, in dimension $d=2$, if $M$ is a set of edges such that $\flip(M+t)$ with
%$K=\langle(1,1),(2,-5)\rangle$ and $t=(1,3)$, then $a_1=b_2=5$, $a_2=b_1=2$ and
%$M=-\G_{(2,5)}^c=\G_{(5,2)}$ (see Figure~\ref{fig:christo25compinv}).
%
%\begin{figure}[h!]
%\begin{center}
%\includegraphics{christo_2_5_comp_inv.pdf}
%\end{center}
%\caption{In general, the equality
%$\G_{(a_1,a_2)}=-\G_{(a_2,a_1)}^c$ holds. The figure
%illustrates that $\G_{(5,2)}=-\G_{(2,5)}^c$. }
%\label{fig:christo25compinv}
%\end{figure}

\begin{corollary}
Let $d=3$.
%If $M=\flip(M+t)$, then 
$M$ is the Christoffel graph $\Ha$
or $M$ is the complement of the reversal of the Christoffel graph $\Hb$.
\end{corollary}

\noindent
Note that the complement of the reversal is equal to the reversal of the
complement.
\begin{proof}
From Lemma~\ref{lem:existsdivisor} there exist integers $0<q<3$ and $0<\ell<3$
such that $\omega\cdot q=a_1+a_2+a_3$ and $\omega\cdot \ell=b_1+b_2+b_3$.
Therefore, there are two cases, either $q=1$ and $\ell=2$ or $q=2$ and
$\ell=1$.
If $q=1$, then Theorem~\ref{thm:Mischristo} (i) applies. Therefore,
$M$ is a Christoffel graph $M=\Ha$ for the vector $\a=(a_1, a_2, a_3)$.
If $\ell=1$, then Theorem~\ref{thm:Mischristo} (ii) applies. Therefore,
the complement of $M$ is the reversal of a
Christoffel graph. More precisely, $M^c=-\Hb$ for the vector
$\b=(b_1, b_2, b_3)$.
\end{proof}

The previous result has also a counterpart in the triangular lattice $L$.

\begin{corollary}
Let $M'\subset \pi(\bE_d)$ such that $\flip(M'+t')=M'$ for some $t'\in L$, that is
invariant under some subgroup of finite index of $L$ and that satisfies
$\Qzero \cap M=\{(0,h_1), (0,h_2), \ldots, (0,h_d)\}+K$. If $d=3$, then $M'$
is equal to a graph $\Ia$ or to the reversal of its edge-complement.
\end{corollary}

\begin{proof}
All we have to do is to lift $M'$ to a set $M\subset \bE_d$ using the
projection $\pi:\R^d\to \D$ in such a way that $M=\flip(M+t)$, with $\pi(t)=t'$, and to show that $M$ is invariant under the subgroup $K=\pi^{-1}(K')$ and satisfies 
$\Qzero \cap M=\{(0,e_1), (0,e_2), \ldots, (0,e_d)\}+K$. Then the corollary follows from the previous one. The details are left to the reader.
\end{proof}

Finally, we give the similar result for Christoffel parallelograms. 
We consider some parallelogram $P$ whose vertices are in $L$, and edges (which
are in $\bE'_d$) within it; these edges must be torally compatible, in the
sense that such an edge hits some edge of the parallelogram, then it reappears
on the opposite edge of the parallelogram. Such a parallelogram defines a
subgroup of finite index $K'$ of $L$ (spanned by the edges of the
parallelogram) and tiles the whole hyperplane $D$. We say that
$\flip(P)=P+t'$, for some $t'\in L$, if $P+t'$ is the parallelogram obtained
by flipping the edges of $P$ incident to zero mod $K$.

\begin{corollary}
Under the previous hypothesis, $P$ is a Christoffel parallelogram or the reversal of its edge-complement.
\end{corollary}

\begin{proof}
It suffices to verify that $P$ defines a subset $M'$ of $\pi(\bE_d)$ satisfying the hypothesis of the previous corollary.
\end{proof}

The result is illustrated in 
Figure~\ref{fig:christo378}.

\begin{figure}[h!]
\begin{center}
%\begin{minipage}[c]{0.18\linewidth}
%\begin{tabular}{l}
%$t=e_3-e_2$ \\
%$\omega=18$ \\
%$\a=(3,7,8)$ \\
%$\b=(15,11,10)$ \\
%$\sum a_i=18=\omega$ \\
%$\sum b_i=36=2n$ \\
%$M=\Ga$
%\end{tabular}
%\end{minipage}
%\begin{minipage}[c]{0.6\linewidth}
\begin{tabular}{cc}
\includegraphics[width=0.3\linewidth]{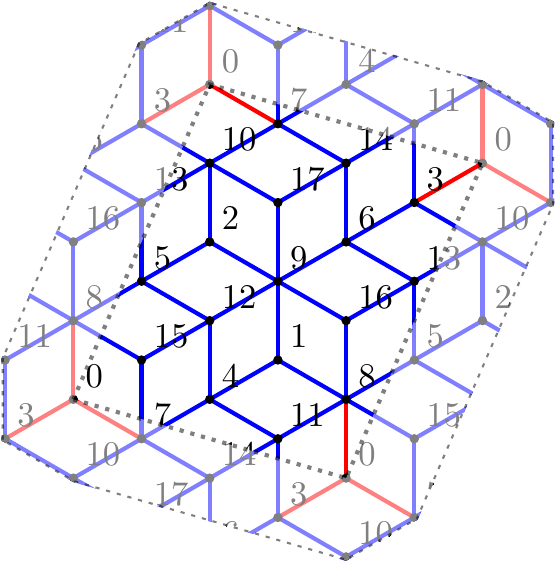} &
\includegraphics[width=0.3\linewidth]{christoffel_15_11_10_mod18.pdf}\\
$\H_{(3,7,8)}$ &
$\H_{(15,11,10),18}$
\end{tabular}
%\end{minipage}
%\begin{minipage}[c]{0.18\linewidth}
%\begin{tabular}{l}
%$t=e_2-e_3$ \\
%$\omega=18$ \\
%$\a=(15,11,10)$\\
%$\b=(3,7,8)$\\
%$\sum a_i=36=2n$ \\
%$\sum b_i=18=\omega$ \\
%$M=\bE_d\setminus-\Gb$
%\end{tabular}
%\end{minipage}
\end{center}
\caption{Left: the Christoffel graph $\Ha$ for the vector
$\a=(3,7,8)$. It satisfies the equation $M=\flip(M+t)$
for the translation vector $t=e_3-e_2$. Right: the complement of
the reversal of the Christoffel graph for the vector $\b=(3,7,8)$, i.e.
$-\Hb^c$. 
It corresponds to the Christoffel graph $\Hao$ for the vector
$\a=(15,11,10)$ and width $\omega=18$.
It satisfies the equation $M=\flip(M+t)$
for the translation vector $t=e_2-e_3$. They represent the only two
possibilities for a pattern $M$ satisfying $M=\flip(M+t)$ when $d=3$ and
$K$ is the subgroup of $\Z^3$ given by $\langle(0,4,1), (-2,0,3), (1,1,1)\rangle$.
}
\label{fig:christo378}
\end{figure}

%illustrates a case where
%Theorem~\ref{thm:Mischristo} applies for vector $\a=(3,7,8)$ (left) and where
%Theorem~\ref{thm:reversalcomplementofMischristo} applies for vector
%$\b=(3,7,8)$ (right).
%For reference and comparison, the Christoffel $\Ga$ of
%normal vector $\a=(15,11,10)$ is illustrated at
%Figure~\ref{fig:christo151110}.
%
%\begin{figure}[h!]
%\begin{center}
%\begin{tabular}{c}
%\includegraphics{christoffel_15_11_10.pdf} \\
%$\G_{(15,11,10)}$
%\end{tabular}
%\end{center}
%\caption{The Christoffel graph $\Ga$ of normal vector $\a=(15,11,10)$.}
%\label{fig:christo151110}
%\end{figure}
%

%An example is shown at Figure~\ref{fig:christo378}. When $d=3$, contrary to
%when $d=2$, it is not true that $-\G_{(b_1,b_2,b_3)}^c=\Ga$ is a
%Christoffel graph for some other vector $\a\in\Z^3$.

We are now ready for the main result of this article which generalizes
Pirillo's theorem (Theorem~\ref{thm:pirillo}) to arbitrary dimension: a graph
$M\subseteq\bE_d$ is a translate of its flip if and only if it is a
Christoffel graph of width $\omega$.

\begin{theorem}
[\bf $d$-dimensional Pirillo's theorem]
\label{thm:pirilloforalld}
Let $K$ be a subgroup of finite index of $\Z^d$.
Let $M\subseteq\bE_d$ be a subset of edges invariant for the group of
translations $K$ such that the
edges of $M$ incident to zero mod $K$ are
$\Qzero\cap M=\{(0,e_i)\mid 1\leq i\leq d\}+K$.
There exists $t\in\Z^d$ such that $M=\flip(M+t)$ if and only if $M=\Hao$ is
the Christoffel graph of normal vector $\a$ and width $\omega$.
\end{theorem}

\begin{proof}
Suppose $M=\flip(M+t)$ for some $t\in\Z^d$.
From Lemma~\ref{lem:generatedbyt},
for all $u=(x_1,x_2,\cdots,x_d)\in\Z^d$ there exists an integer $k$ such that
$u=kt$ with
\[
k
= \sum (-b_ix_i)\bmod \omega
= \sum (a_i-\omega)x_i\bmod \omega
= \sum a_ix_i\bmod \omega
= \Fao(u).
\]
We want to show that $M=\Hao$.
We have that $(u,u+e_i)=(0,e_i)+kt\in M$ if and only if $0\leq k<b_i$ if and
only if $\Fao(u)\in[0,\omega-a_i-1]$ if and only if $(u,u+e_i)\in\Hao$ from
Lemma~\ref{lem:edgesintervalforHan}.

Reciprocally,
suppose $\Hao$ is
the Christoffel graph of normal vector $\a$ of width $\omega$.
We can show that $\Hao+t=\flip(\Hao)$ where $t\in\Z^d$ is such that
$\Fao(t)=1$.
The proof goes along the same lines as Proposition~\ref{prop:translation}
using Lemma~\ref{lem:edgesintervalforHan} instead of
Lemma~\ref{lem:edgesinterval}. 
\end{proof}

%\section{Conclusion}

\section{Appendix: Discrete planes}\label{sec:observ}

In this section, we show some results on standard discrete planes.  Discrete
planes were introduced in \cite{Reveilles_1991} and standard discrete planes
were further studied in \cite{MR1382845}.  The projection of a standard
discrete plane gives a tiling of $\D$ by three kinds of rhombus
\cite{MR1782038} thus yielding a coding of it by $\Z^2$-actions by rotations on
the unit circle \cite{MR1906478,MR2330996}.
Our construction of the discretized hyperplane is equivalent, for the dimension
3, to that in \cite{MR1906478}. 
Our point of view is slightly different from the classical one;
%(Rev91, Fra96, ABI02, ABFJ07); 
inspired by the 2-dimensional case (discrete lines), we
define a discrete hyperplane by ``what the observer sees": the observer is at
$-\infty$ in the direction $(1,1,\ldots,1)$ and he looks towards the
``transparent" hyperplane all the unit hypercubes which are located on the
other side. This may be modelled mathematically; all the results are
intuitively clear, but require a proof. We prove them, since we could not find
precise references. We recover some known results.

Imagine the $d$-dimensional space filled with unit hypercubes with opaque faces.
Consider a transverse hyperplane generated by its integer points (formally, of
equation $\sum a_i x_i=0$, $a_i>0$ coprime integers). As an observer, we install
in the open half-space $H_-$ bounded by the plane. Then, we remove all the cubes in this half-space,
including the cubes intersecting this half-space; in other words, we keep only
the cubes contained in $H_+$. 
Figure~\ref{fig:graph25observer} illustrates this construction for $d=2$.
\begin{figure}[h]
\begin{center}
\includegraphics{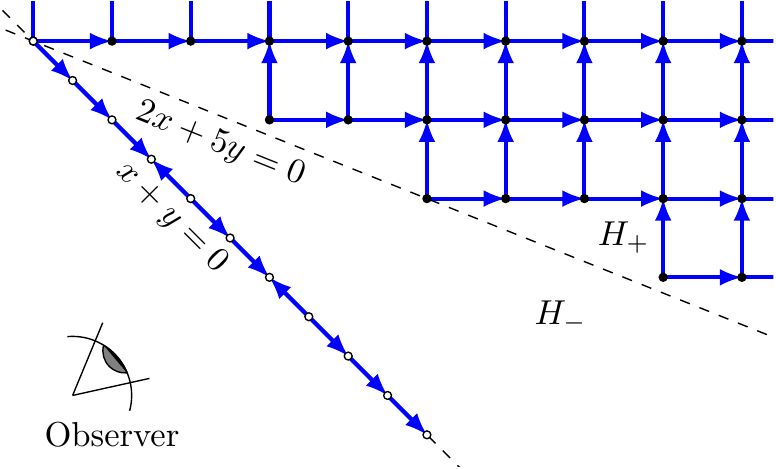}
\end{center}
\caption{Observation in dimension 2. What the observer sees can be projected
 parallel to the vector $(1,1)$ on the line $x+y=0$.}
\label{fig:graph25observer}
\end{figure}
For $d=3$, when we look towards $H_+$ parallely to the vector $(1,1,1)$, then
we see something like in Figure~\ref{fig:graph235}.
\begin{figure}[h]
\begin{center}
\includegraphics{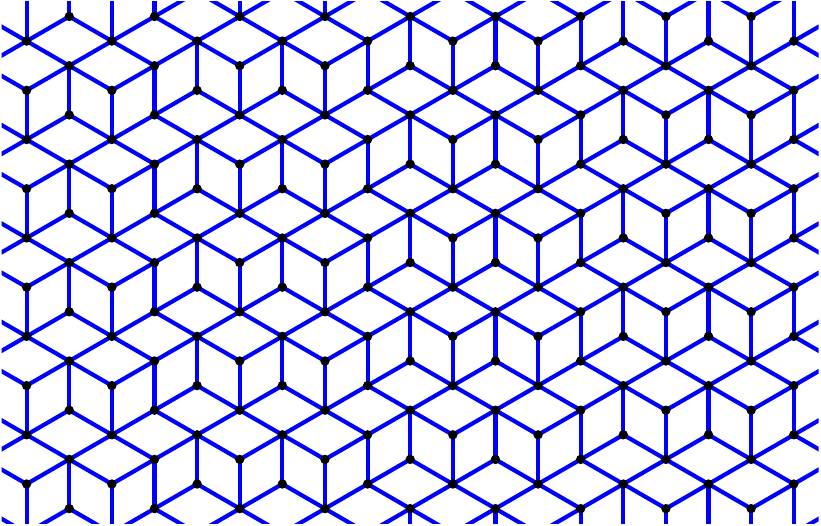}
\end{center}
\caption{What the observer sees in dimension 3. The surface of cubes was projected
parallel to the vector $(1,1,1)$ on the plane $x+y+z=0$.}
\label{fig:graph235}
\end{figure}
%\todo{change this example to a prime order}
%\todo{choose a normal vector such that the sum is coprime with the entries}

Let $s$ be the sum $s=\sum_i a_i$.
We denote $\a=(a_1,a_2,\ldots,a_d) \in \Z^d$.
The complement of $H$ has two connected components $H_-$ and $H_+$, where the first is determined by the inequation $\sum_ia_ix_i<0$.
We consider the unit cubes in $\R^d$ and their facets. Such a facet is a subset of $\R^d$ of the form $M+\sum_{j\neq i} [0,1]e_j$, for some coordinate $i\in \{1,\ldots,d\}$ and some integral point $M\in \Z^d$. Denote by $\cal C$ the standard unit cube.

Consider the unit hypercubes that are contained in the closed half-space $H\cup H_+$ and their facets; denote by $\cal U_+$ the union of all these facets.
Note that a unit cube $M+\cal C$ ($M\in \Z^d$) is contained in $H\cup H_+$ if and only if $M\in H\cup H_+$ if and only if $\sum_ja_jm_j\geq 0$.
%Thus a facet $M+\sum_{j\neq i} [0,1]e_j$ is in $\cal U_+$ if and only if $\sum_ja_jm_j\geq 0$.
%Denote also by $\cal U_-$ the union of all remaining  facets. 

We say that a point $M$ in $\R^d$ is {\em visible} if the open half-line $M+]-\infty,0[ (1,1,\ldots,1)$ does not contain any point in $\cal U_+$. Intuitively, this means that, all facets being opaque, 
%one has removed all of $\cal U_-$, and 
that an observer located at infinity in the direction of the vector $-(1,1,\ldots,1)$ can see this point $M$, because no point in $\cal U_+$ hides this point. 

Now, we consider the set of visible points which belong to $\cal U_+$. This we may call the {\em discretized hyperplane associated to } $H$. Intuitively, it is the set of facets that the observer can sees, 
%after removal of $\cal U_-$, and 
as is explained in the introduction.

We characterize now the discretized hyperplane. For this, we denote by $R$ the following subset of $\R^d$: $R=\{(x_i) \mid 0\leq \sum_ia_ix_i <s\}$.
Note that $R\subset H\cup H_+$.

Denote by $\S$ the union of the facets that are contained in $R$. In other
words, $${\S}=\bigcup_{M\in \Z^d,1\leq i\leq d, M+\sum_{j\neq i} [0,1]e_j\subset R} (M+\sum_{j\neq i} [0,1]e_j).$$ 

Observe that the condition $M+\sum_{j\neq i} [0,1]e_j\subset R$ is equivalent to: $\sum_ja_jm_j\geq 0$ and $a_im_i+\sum_{j\neq i}a_j(m_j+1) <s$.
Note also that $\S \subset \cal U_+$.

\begin{theorem}
The discretized hyperplane is equal to $\S$.
\end{theorem}

Observe that if we project $\S$ onto the hyperplane perpendicular to the vector $(1,1,\ldots,1)$, we obtain exactly what the observer sees.
%; this projection is a bijection from $\S$ onto the hyperplane. 

An example of this, for $d=3$, is given in Figure~\ref{fig:graph235}.
This observation motivates the introduction of the graph $\Ia$ in
Section~\ref{sec:Ia}.

We first give an simple characterization of $\S$.

\begin{proposition}\label{characF}
Let $X=(x_i)\in \R^d$. Then $X$ is in $\S$ if and only if the three conditions below hold:

(i) some coordinate of $X$ is an integer;

(ii)$ \sum_i a_i \lfloor x_i \rfloor \geq 0$;

(iii) $\sum_i a_i \lceil x_i \rceil <s$.
\end{proposition}

We recover Proposition 1 of \cite{MR1906478}. 
\begin{corollary}
    Let $X=(x_i)\in \Z^d$. Then $X$ is in $\S$ if and only if $0\leq \sum_ia_ix_i<s$.
\end{corollary}

\begin{proof}[Proof of the proposition.]
Suppose that $X\in \S$. Then $X\in M+\sum_{j\neq i} [0,1]e_j\subset \S$ and the coordinates $m_j$ of $M$ are integers. Thus, by an observation made previously,  $0\leq \sum_ja_jm_j\leq \sum_ja_j\lfloor x_j \rfloor$, since $x_j=m_j+\theta_j$, with $0\leq\theta_j\leq 1$ and $\theta_i=0$. Moreover, $\lceil x_i \rceil=m_i$, and $\lceil x_j \rceil\leq m_j+1$ if $j\neq i$. Thus, $ \sum_j a_j \lceil x_j \rceil \leq a_im_i+\sum_{j\neq i}(m_j+1) <s$, by the same observation.

Conversely, suppose that the three conditions of the proposition hold. Without restricting the generality (by 
permutation of the coordinates), we may assume that for some $i\in \{1,\ldots,d\}$, one has $x_i,
\ldots,x_i\in \Z$ and $x_{i+1},\ldots,x_d\notin \Z$. Let $0\leq p\leq i$ be maximum subject to the condition $\sum_{j\leq p}a_j(x_j-1)+\sum_{j> p}a_j\lfloor x_j\rfloor \geq 0$ (note that $p$ exists since the inequality is satisfied for $p=0$). Suppose that $p=i$; then $\sum_j a_j \lceil x_j \rceil =\sum_{j\leq i}a_jx_j+\sum_{j>i}a_j(\lfloor x_j \rfloor+1)=a_1+\cdots+a_d+\sum_{j\leq i}a_j(x_j-1)+\sum_{j>i}a_j\lfloor x_j \rfloor \geq a_1+\cdots+a_d$ (since $p=i$) $=s$; thus we obtain a contradiction with condition (iii).

Thus $p<i$ and $p+1\leq i$. We have by maximality the inequality $\sum_{j\leq p+1}a_j(x_j-1)+\sum_{j> p+1}a_j\lfloor x_j\rfloor < 0$.
 %We deduce that $\sum_{j\leq p}a_j(x_j-1)+a_{p+1}x_{p+1}+\sum_{j> p+1}a_j\lfloor x_j\rfloor < a_{p+1}$. 
 Let $M=(m_j)=(x_1-1,\ldots,x_p-1,\lfloor x_{p+1} \rfloor,\ldots,\lfloor x_d \rfloor \in \Z^d$. We have $\sum_ja_jm_j\geq 0$ (by definition of $p$) and $a_{p+1}m_{p+1}+\sum_{j\neq p+1}a_j(m_j+1)=\sum_{j\leq p}(a_j(x_j-1)+a_j)+a_{p+1}(x_{p+1}-1)+a_{p+1}+\sum_{j> p+1}(a_j\lfloor x_j\rfloor+a_j)
 =s+\sum_{j\leq p+1}a_j(x_j-1)+\sum_{j> p+1}a_j\lfloor x_j\rfloor  <s$, by the previous inequality.
 %a_{p+1}x_{p+1}+\sum_{j\leq p}a_j(x_j-1)+\sum_{j\leq p}a_j+\sum_{j> p+1}a_j\lfloor x_j\rfloor+\sum_{j> p+1}a_j<a_{p+1}+\sum_{j\leq p}a_j+\sum_{j> p+1}a_j=s$. 
 Thus $M+\sum_{j\neq i} [0,1]e_j\subset \S$, by the observation made above. Moreover $X\in M+\sum_{j\neq i} [0,1]e_j$ since $p+1\leq i$.
\end{proof}

\begin{corollary}\label{XY}
For each point $X$ in $\R^d$, there is a unique point $Y$ in $\S$ such that $XY$ is parallel to the vector $(1,1,\ldots,1)$.
\end{corollary}

Denote by $f$ the function such that $Y=f(X)$ with the notations of the
corollary. This function is a kind of projection onto $\S$, parallely to the
vector $(1,1,\ldots,)$. Denote also by $t(X)$ the real-valued function defined
by $X=f(X)+t(X)(1,1,\ldots,1)$, and by $t(X)=0$ if and only if $X\in \S$. 

\begin{proof}
We prove first unicity. By contradiction: we have $Y,Z \in \S$ and $Z=Y+t(1,1,\ldots,1)$ with $t>0$. Then $z_i=y_i+t$. Thus 
$\lceil z_i \rceil \geq \lfloor y_i \rfloor +1$. Hence $\sum_i a_i \lceil z_i \rceil \geq s+\sum_i a_i \lfloor y_i \rfloor$. Since by the proposition, applied to $Y$, the last sum is $\geq 0$, we obtain $\sum_i a_i \lceil z_i \rceil \geq s$, which contradicts the proposition, applied to $Z$.

We prove now the existence of $Y$. We may assume that $L(X)= \sum_i a_i \lfloor x_i \rfloor \geq 0$, by adding to $X$ some positive multiple of $(1,1,\ldots,1)$ if necessary. We prove existence of $Y$ by induction on the sum $U(X)=\sum_i a_i \lceil x_i \rceil $.

Let $\epsilon=\min_i(x_i-\lfloor x_i \rfloor)$. Observe that if we replace $X$ by $X-\epsilon (1,1,\ldots,1)$, then $L(X)$ does not change, $U(X)$ does not increase and moreover some $x_i$ is now an integer.

If $U(X)$ is $<s$, this observation implies the existence of $Y$. 

Suppose now that $U(X)\geq s$. By the observation, we may assume that at least one of the $x_i$ is an integer. Without restricting the generality, we may also assume that $x_1,\ldots,x_i\in \Z$ and that $x_{i+1},\ldots,x_d\notin \Z$, with $i\geq 1$. 

If $i=d$, then the $x_j$ are all integers, $L(X)=U(X)$, we replace $X$ by $X-(1,1,\ldots,1)$ and we conclude by induction, since $L(X)$ is replaced by $L(X)-s$.

Suppose now that $i<d$. Let $\epsilon=\min_{j>i}(x_j-\lfloor x_j \rfloor)$; then $\epsilon >0$. We have $s\leq \sum_j a_j \lceil x_j \rceil  = \sum_{j\leq i} a_j x_j  + \sum_{j>i} a_j ( \lfloor x_j \rfloor +1)=L(X)+a_{i+1}+\ldots+a_d$, hence $L(X)\geq a_1+\cdots +a_i$. Note that $\sum_j a_j(\lfloor x_j-\epsilon \rfloor)= \sum_{j\leq i} a_j (x_j-1)+
\sum_{j>i} a_j  \lfloor x_j \rfloor = L(X)-a_1-\cdots-a_i \geq 0$. We replace $X$ by $X-\epsilon (1,1,\ldots,1)$, and we may conclude by induction, since $U(X)$ strictly decreases and since $L(X)$ remains $\geq 0$.
\end{proof}

\begin{proof} (of the theorem)
Let $X$ be a point 
on the discretized hyperplane associated to $H$.
Suppose that $t(X)>0$. Then
$X=f(X)+t(X)(1,1,\ldots,1)$ so that $X$ is hidden by $f(X)$: formally, $f(X)$
is on the open half-line $X+]-\infty,0[ (1,1,\ldots,1)$ and since $f(X)$ is in
$\S$, it is a point in $\cal U_+$. We conclude that we must have $t(X)\leq 0$.
Suppose that $t(X)<0$. We know that $X$ is in $\cal U_+$, so that $X$ belongs
to a hypercube $M+\cal C$ with $\sum_ja_jm_j\geq 0$, and therefore $x_j\geq
m_j$. Let $Y=f(X)$. Then $X=Y+t(X)(1,1,\ldots,1)$ so that $y_j>x_j\geq m_j$
which implies $\sum_ja_j\lceil y_j\rceil\geq \sum_ja_j(m_j+1)\geq s$, a
contradiction with Proposition \ref{characF}. Thus $t(X)=0$ and $X\in \S$.

Conversely suppose that $X\in \S$. Suppose that $X$ is not 
on the discretized hyperplane associated to $H$.
This implies that there is some point $Y\in \cal U_+$ on the open half-line $X+]-\infty,0[ (1,1,\ldots,1)$. We have $Y\in M+\cal C$ with $\sum_ja_jm_j\geq 0$. Thus $x_j>y_j\geq m_j$ which implies that $\sum_ja_j\lceil x_j\rceil\geq \sum_ja_j(m_j+1)\geq s$, a contradiction with Proposition \ref{characF}.
\end{proof}

\begin{corollary}
Let $d\geq 2$. Let $M\in {\S}\cap \Z^d$. Let $i=1,2,\ldots,d$ and $N=M+e_i$.

(i) $N\in \S$ if and only if $\sum_ja_j n_j<s$; in this case, the segment
$M+[0,1]e_i$ is contained in $\S$. 

(ii) If $N\notin \S$, then the only point in $(M+[0,1]e_i) \cap{\S} $ is $M$.
\end{corollary}

\begin{proof} The fact that $N\in \S$ if and only if $\sum_ja_j n_j<s$ is a consequence of the proposition.

Suppose that $N\in \S$ and let $X$ be on the segment $M+[0,1]e_i$. Then $0\leq  \sum_j a_j m_j\leq\sum_ja_j \lfloor x_j \rfloor $ and $\sum_ja_j \lceil x_j \rceil \leq \sum_ja_jn_j<s$. Thus the corollary follows from the proposition.

Suppose no that $N\notin \S$ and let $X$ be on this segment. Since  $0\leq
\sum_j a_j m_j $, we have also $0\leq  \sum_j a_j n_j$. Since $N\notin \S$, we must have $\sum_ja_jn_j\geq s$. Moreover, if $X\neq M$, we have
$\lceil x_j \rceil=n_j$, so that $\sum_ja_j \lceil x_j \rceil\geq s$ and
$X\notin \S$.
\end{proof}

The next result, which is not needed in this article, is of independent
interest, and intuitively clear (but it requires a proof).

\begin{proposition}
The function $f: X\mapsto Y$, with the notations of Corollary \ref{XY}, is
continuous. The open set $\R^d\setminus\S$ has two connected components.
\end{proposition}

\begin{lemma} 
Let $\S$ be a closed  subset of $\R^d$such that for each $X$ in $\R^d$, there
is a unique $Y$ in $\S$ such that $XY$ is parallel to $(1,1,\ldots,1)$. If the mapping $X\mapsto Y$ is bounded, then it is continuous.
\end{lemma} 

\begin{proof} Recall that a bounded sequence in $\R^d$, converges if for any two convergent 
subsequences, they have the same limit. Let $(X_n)$ be a sequence in $\R^d$, with limit $l$. It is enough to show that $(f(X_n))$ converges; note that this sequence is bounded. Consider 
two subsequences of $(X_n) $ such that their images under $f$ have limits, $l_1$ and $l_2$ say. Since 
$\S$ is closed, $l_1,l_2\in \S$. Let $\epsilon >0$. For $n$ large enough, $\mid X_n-l\mid <\epsilon$; hence $f(X_n)$ is in the open cylinder of diameter $\epsilon$ and with axis the line $l+(1,1,\ldots,1)$. This implies that $l_1,l_2$ are in this cylinder and consequently, $\epsilon$ being arbitrary, $l_1,l_2$ are on the previous line. By unicity, $l_1=l_2$ ($=f(l)$). We conclude using the remark at the beginning of the proof.
\end {proof} 

\begin{proof} (of the proposition)
The mapping $f$ is continuous: by the lemma, it is enough to show that $\S$ is
closed and that the mapping is bounded. Since each convergent sequence is
contained in some compact set, it is enough to show that for each compact set
$K$, $K\cap \S$ is closed; but this is clear, since the latter set is the union of finitely many $K\cap F$, $F$ facet of unit hypercube. The mapping is bounded since its image is between the two hyperplanes of equations $\sum_ia_ix_i=0$ and $\sum_ia_ix_i=s$, so that the image of each bounded set is contained in a cylinder of axis parallel to $(1,1,\ldots)$ and limited by these two hyperplanes.

Now, we show that the  set $\R^d\setminus\S$ has two connected components.
Note that for each point $X$, one has $X=f(X)+t(X)(1,1,\ldots,1)$ for some
continuous real-valued function $t$. Since
$f(1,1,\ldots,1)=(0,0,\ldots,0)=f(-1,-1,\ldots,-1)$, one has
$t(1,1,\ldots,1)=1$ and $t(-1,-1,\ldots,-1)=-1$. Moreover $t(X)=0$ if and only
if $X\in \S$. Thus $t(\R^d\setminus \S)$ is not 
connected and neither is $\R^d\setminus \S$. 

Now, if $t(X)>0$, one may connect $X$ by a piece of the line
$X+\R(1,1,\ldots,1)$ to a point of the half-space $\sum_ia_ix_i>0$ and this
implies that the set of points $X$ with $t(X)>0$ is connected. Similarly, the
set of points with $t(X)<0$ is connected, and $\R^d)\setminus \S$ has therefore two connected components.
\end{proof}

We recover Proposition 2 of \cite{MR1906478} and Proposition 4 of \cite{MR2330996}.

\begin{corollary}
The restriction of $f$ to $\D$ is a homeomorphism of $\D$ onto $\S$.
\end{corollary}

\begin{proof}
Indeed, the inverse mapping is the projection onto the hyperplane $\D$
parallely to the vector $(1,1,\ldots,1)$.
\end{proof}

%%%%%%%%%%%%%%%%
% Bibliographie %
%%%%%%%%%%%%%%%%%
\bibliographystyle{alpha}
\bibliography{biblio}

\begin{thebibliography}{BHNR04}

\bibitem[ABEI01]{MR1888763}
Pierre Arnoux, Val{\'e}rie Berth{\'e}, Hiromi Ei, and Shunji Ito.
\newblock Tilings, quasicrystals, discrete planes, generalized substitutions,
  and multidimensional continued fractions.
\newblock In {\em Discrete models: combinatorics, computation, and geometry
  ({P}aris, 2001)}, Discrete Math. Theor. Comput. Sci. Proc., AA, pages
  059--078 (electronic). Maison Inform. Math. Discr\`et. (MIMD), Paris, 2001.

\bibitem[ABFJ07]{MR2330996}
Pierre Arnoux, Val{\'e}rie Berth{\'e}, Thomas Fernique, and Damien Jamet.
\newblock Functional stepped surfaces, flips, and generalized substitutions.
\newblock {\em Theoret. Comput. Sci.}, 380(3):251--265, 2007.

\bibitem[ABI02]{MR1906478}
Pierre Arnoux, Valerie Berth{\'e}, and Shunji Ito.
\newblock Discrete planes, {$\Bbb Z^2$}-actions, {J}acobi-{P}erron algorithm
  and substitutions.
\newblock {\em Ann. Inst. Fourier (Grenoble)}, 52(2):305--349, 2002.

\bibitem[ARC97]{A}
Eric Andres, Acharya Raj, and Sibata Claudio.
\newblock Discrete analytical hyperplanes.
\newblock {\em Graphical Models and Image Procesing}, 59(5):302--309, 1997.

\bibitem[BCK07]{MR2296869}
Valentin Brimkov, David Coeurjolly, and Reinhard Klette.
\newblock Digital planarity---a review.
\newblock {\em Discrete Appl. Math.}, 155(4):468--495, 2007.

\bibitem[Ber07]{berstel_sturmian_2007}
Jean Berstel.
\newblock Sturmian and episturmian words.
\newblock In Symeon Bozapalidis and George Rahonis, editors, {\em Algebraic
  Informatics}, volume 4728 of {\em Lecture Notes in Computer Science}, pages
  23--47. Springer Berlin Heidelberg, 2007.

\bibitem[BFR08]{MR2440650}
Olivier Bodini, Thomas Fernique, and {\'E}ric R{\'e}mila.
\newblock A characterization of flip-accessibility for rhombus tilings of the
  whole plane.
\newblock {\em Inform. and Comput.}, 206(9-10):1065--1073, 2008.

\bibitem[BFRR11]{MR2856174}
Olivier Bodini, Thomas Fernique, Michael Rao, and {\'E}ric R{\'e}mila.
\newblock Distances on rhombus tilings.
\newblock {\em Theoret. Comput. Sci.}, 412(36):4787--4794, 2011.

\bibitem[BHNR04]{bhnr}
Srecko Brlek, Sylvie Hamel, Maurice Nivat, and Christophe Reutenauer.
\newblock On the palindromic complexity of infinite words.
\newblock {\em Int. J. Found. Comput. Sci.}, 15(2):293--306, 2004.

\bibitem[BLRS08]{BLRS08}
Jean Berstel, Aaron Lauve, Christophe Reutenauer, and Franco Saliola.
\newblock {\em Combinatorics on Words: Christoffel Words and Repetition in
  Words}, volume~27 of {\em CRM monograph series}.
\newblock American Mathematical Society, 2008.
\newblock 147 pages.

\bibitem[BR06]{MR2197281}
Jean-Pierre Borel and Christophe Reutenauer.
\newblock On {C}hristoffel classes.
\newblock {\em Theor. Inform. Appl.}, 40(1):15--27, 2006.

\bibitem[BT04]{MR2074953}
Val{\'e}rie Berth{\'e} and Robert Tijdeman.
\newblock Lattices and multi-dimensional words.
\newblock {\em Theoret. Comput. Sci.}, 319(1-3):177--202, 2004.

\bibitem[BV00]{MR1782038}
Val{\'e}rie Berth{\'e} and Laurent Vuillon.
\newblock Tilings and rotations on the torus: a two-dimensional generalization
  of {S}turmian sequences.
\newblock {\em Discrete Math.}, 223(1-3):27--53, 2000.

\bibitem[Chu97]{MR1483437}
Wai-Fong Chuan.
\newblock {$\alpha$}-words and factors of characteristic sequences.
\newblock {\em Discrete Math.}, 177(1-3):33--50, 1997.

\bibitem[CL05]{carpi_central_2005}
Arturo Carpi and Aldo Luca.
\newblock Central sturmian words: Recent developments.
\newblock In Clelia Felice and Antonio Restivo, editors, {\em Developments in
  Language Theory}, volume 3572 of {\em Lecture Notes in Computer Science},
  pages 36--56. Springer Berlin Heidelberg, 2005.

\bibitem[{Deb}95]{debled-rennesson_reconnaissance_1995}
Isabelle {Debled-Rennesson}.
\newblock {\em Reconnaissance des droites et plans discrets}.
\newblock {Th\`ese de Doctorat}, Université Louis Pasteur - Strasbourg, 1995.

\bibitem[DV12]{MR3052947}
Eric Domenjoud and Laurent Vuillon.
\newblock Geometric palindromic closure.
\newblock {\em Unif. Distrib. Theory}, 7(2):109--140, 2012.

\bibitem[{Fer}07]{fernique_these_2007}
Thomas {Fernique}.
\newblock {\em Pavages, fractions continues et géométrie discrète}.
\newblock {Th\`ese de Doctorat}, Université Montpellier 2, 2007.

\bibitem[Fra96]{MR1382845}
Jean Fran{\c{c}}on.
\newblock Sur la topologie d'un plan arithm\'etique.
\newblock {\em Theoret. Comput. Sci.}, 156(1-2):159--176, 1996.

\bibitem[FST96]{MR1603656}
Jean Fran{\c{c}}on, Jean-Maurice Schramm, and Mohamed Tajine.
\newblock Recognizing arithmetic straight lines and planes.
\newblock In {\em Discrete geometry for computer imagery ({L}yon, 1996)},
  volume 1176 of {\em Lecture Notes in Comput. Sci.}, pages 141--150. Springer,
  Berlin, 1996.

\bibitem[IO93]{MR1247666}
Shunji Ito and Makoto Ohtsuki.
\newblock Modified {J}acobi-{P}erron algorithm and generating {M}arkov
  partitions for special hyperbolic toral automorphisms.
\newblock {\em Tokyo J. Math.}, 16(2):441--472, 1993.

\bibitem[IO94]{MR1279568}
Shunji Ito and Makoto Ohtsuki.
\newblock Parallelogram tilings and {J}acobi-{P}erron algorithm.
\newblock {\em Tokyo J. Math.}, 17(1):33--58, 1994.

\bibitem[PBDR06]{MR2305655}
Laurent Provot, Lilian Buzer, and Isabelle Debled-Rennesson.
\newblock Recognition of blurred pieces of discrete planes.
\newblock In {\em Discrete geometry for computer imagery}, volume 4245 of {\em
  Lecture Notes in Comput. Sci.}, pages 65--76. Springer, Berlin, 2006.

\bibitem[Pir01]{MR1854493}
Giuseppe Pirillo.
\newblock A curious characteristic property of standard {S}turmian words.
\newblock In {\em Algebraic combinatorics and computer science}, pages
  541--546. Springer Italia, Milan, 2001.

\bibitem[Rev91]{Reveilles_1991}
Jean-Pierre Reveill{\`e}s.
\newblock {\em G\'eom\'etrie discr\`ete, calcul en nombres entiers et
  algorithmique}.
\newblock {Th\`ese de Doctorat}, Universit\'e Louis Pasteur, Strasbourg, 1991.

\bibitem[Rev95]{MR1368203}
Jean-Pierre Reveill{\`e}s.
\newblock Combinatorial pieces in digital lines and planes.
\newblock In {\em Vision geometry, {IV} ({S}an {D}iego, {CA}, 1995)}, volume
  2573 of {\em Proc. SPIE}, pages 23--34. SPIE, Bellingham, WA, 1995.

\bibitem[Vui99]{MR1732895}
Laurent Vuillon.
\newblock Local configurations in a discrete plane.
\newblock {\em Bull. Belg. Math. Soc. Simon Stevin}, 6(4):625--636, 1999.

\end{thebibliography}

\end{document}